\documentclass[journal]{IEEEtran}
\usepackage{fancyhdr}
\usepackage{times,amsfonts,bbm,dsfont,amsmath,color,amssymb,graphicx,epsfig,multirow,float,algorithm,algorithmic,bm,cite,setspace,amsthm}
\usepackage[normalem]{ulem}
\usepackage{threeparttable}
\usepackage{makecell}
\usepackage{amsmath}
\usepackage{cuted}
\usepackage{stfloats}
\usepackage{arydshln}
\usepackage{psfrag}
\usepackage{epstopdf}
\usepackage{stfloats}
\usepackage{color,cases}
\usepackage{subfigure}
\usepackage[font={normalsize}]{caption}
\newtheorem{myTheo}{Theorem}

\newtheorem{myCor}{Corollary}
\newtheorem{myDef}{Definition}
\newtheorem{myLem}{Lemma}
\newtheorem{myOpt}{Optimization}
\newtheorem{myPro}{Proposition}

\begin{document}

\title{Collaborative Computing in Non-Terrestrial Networks: A Multi-Time-Scale Deep Reinforcement Learning Approach }

\author{\IEEEauthorblockN{Yang Cao, Shao-Yu Lien, Ying-Chang Liang, {\it{Fellow, IEEE}},\\ Dusit Niyato, {\it{Fellow, IEEE}}, and Xuemin (Sherman) Shen, {\it{Fellow, IEEE}}}
\thanks{Part of this work was presented in IEEE PIMRC 2023 \cite{Cao2023}.

Y. Cao is with the School of Information Science and Technology, Southwest Jiaotong University, Chengdu 611756, China. (e-mail: cyang9502@gmail.com)

Y.-C. Liang is with the Center for Intelligent Networking and Communications (CINC), University of Electronic Science and Technology of China, Chengdu 611731, China. (e-mail: liangyc@ieee.org).

S.-Y. Lien is with the Institute of Intelligent Systems, National Yang Ming Chiao Tung University, Tainan City 711, Taiwan. (e-mail: sylien@nycu.edu.tw).

D. Niyato is with the School of Computer Science and Engineering, Nanyang Technological University, Singapore 639798 (e-mail: dniyato@ntu.edu.sg).

X. Shen  is with the Department of Electrical and Computer Engineering, University of Waterloo, 200 University Avenue West, Waterloo, Ontario, Canada, N2L 3G1. (e-mail: sshen@uwaterloo.ca)}}

\maketitle

\thispagestyle{fancy}
\lhead{\small{© 2023 IEEE. Personal use of this material is permitted.  Permission from IEEE must be obtained for all other uses, in any current or future media, including reprinting/republishing this material for advertising or promotional purposes, creating new collective works, for resale or redistribution to servers or lists, or reuse of any copyrighted component of this work in other works.
DOI: 10.1109/TWC.2023.3323554}}
\renewcommand{\headrulewidth}{0pt} 

\begin{abstract}
Constructing earth-fixed cells with low-earth orbit (LEO) satellites in non-terrestrial networks (NTNs) has been the most promising paradigm to enable global coverage. The limited computing capabilities on LEO satellites however render tackling resource optimization within a short duration a critical challenge. Although the sufficient computing capabilities of the ground infrastructures can be utilized to assist the LEO satellite, different time-scale control cycles and coupling decisions between the space- and ground-segments still obstruct the joint optimization design for computing agents at different segments. To address the above challenges, in this paper, a multi-time-scale deep reinforcement learning (DRL) scheme is developed for achieving the radio resource optimization in NTNs, in which the LEO satellite and user equipment (UE) collaborate with each other to perform individual decision-making tasks with different control cycles. Specifically, the UE updates its policy toward improving value functions of both the satellite and UE, while the LEO satellite only performs finite-step rollout for decision-makings based on the reference decision trajectory provided by the UE. Most importantly, rigorous analysis to guarantee the performance convergence of the proposed scheme is provided. Comprehensive simulations are conducted to justify the effectiveness of the proposed scheme in balancing the transmission performance and computational complexity.
\end{abstract}


\section{Introduction}

Toward fulfilling the ubiquitous connectivity scenario of the International Mobile Telecommunications for 2030 (IMT-2030), non-terrestrial networks (NTNs) involving various flying stations (i.e., satellites, high-altitude platforms or unmanned aerial vehicles (UAVs)) have been integrated with new radio (NR) technology of the third generation partnership project (3GPP) to provide wireless services to rural areas uncovered by traditional terrestrial networks \cite{3GPP2021}. In 3GPP NR based NTNs (NR-NTNs), the stations with the transparent payload can only repeat the signals, while the stations with the regenerative payload can further provide coding/modulation, switching/routing and radio resource management (RRM) functions of the BS. From 3GPP Release 15, the geostationary earth-orbit (GEO) and low earth-orbit (LEO) satellites have been regarded as key deployment scenarios \cite{Lin2021}. Particularly, to achieve satellite direct-to-device services, LEO satellites have emerged as a major paradigm with the merits of a lower signal propagation delay and a lower signal attenuation level over GEO satellites. Toward implementing global coverage, more than 17,000 LEO satellites are projected to be launched by 2030 to form constellations. Nevertheless, since LEO satellites revolve around the earth with extremely high speeds (i.e., 7,000 meters per second), the throughput performance of the LEO satellite service is significantly subject to the varying link quality and short dwell duration between the satellite and user equipment (UE) on the ground. Therefore, the moving cell pattern of the LEO satellite raises critical challenges in pursuing stable high-throughput satellite links.

To implement seamless services for the UE, two cell scenarios based on whether an LEO cell is stationary respect to the earth surface (i.e., earth-moving cell and earth-fixed cell) have been considered in 3GPP NR-NTNs \cite{3GPP2021}. In the earth-moving cell scenario, the beam of the LEO satellite fixes toward a certain direction, and thus the beam coverage moves along with the moving position of the LEO satellite. In this case, due to the severe time-varying propagation attenuation incurred by the rapid movement of the LEO satellite, the throughput performance of a specific set of UEs degrades sharply, and these UEs should be handed over frequently between the leaving satellite/beam and a newly arriving satellite/beam. To enhance the consistency of services, in the earth-fixed cell scenario, the LEO satellite needs to adjust its beam direction in a nearly real-time manner to keep its beam pointing a specific terrestrial area within its dwell duration. At the same time, radio resources in the time domain or the frequency domain also need to be optimized timely according to the variations on link qualities and throughput requirements of UEs. To perform the above real-time RRM function, a regenerative payload is mandatory for LEO satellites, which demands considerable on-board computing capabilities. However, the on-board computing capability of the LEO satellite is largely constrained by the energy supply and payload weight, and it may be harmed by cosmic radiations of the space environment in practical deployments. In this case, alleviating the computing burdens of LEO satellites turns out to be the key to empower NTNs.

Fortunately, different from the fact that both LEO satellites and terrestrial sensors have limited energy supplies in 3GPP Narrowband Internet-of-Thing (NB-IoT) based NTNs \cite{Ortigueira2021}, the terrestrial components (i.e., hand-held devices and vehicles) in 3GPP NR-NTNs normally have sufficient energy supplies and can achieve powerful computing capabilities over those non-terrestrial components \cite{Yu2022}. Therefore, to implement the earth-fixed cell scenario, part of huge computations of the RRM function can be offloaded to the UE so that the computing burden on the LEO satellite can be alleviated. To this end, the multi-agent scheme can be adopted as an effective methodology to form collaborative decision-making mechanism to optimize the overall performance, as it can decompose the high-dimensional overall performance optimization into low-dimensional optimizations for multiple agents. However, there are two critical challenges in constructing such a multi-agent scheme involving multiple stations with different computing capabilities. On the one hand, the control cycle for the resource optimization is mainly subject to the computing capability, and thus the control cycles for the LEO satellite and UE with different computing capabilities may be different. In this case, since the resource decisions optimized at different time-scales are coupled, fully distributed agents may be unable to obtain converged policies. On the other hand, although the backward induction spirit in Stackelberg game \cite{Nie2019} is feasible if the knowledge of different stations are perfectly known, the satellite/UE cannot easily obtain its required knowledge from other stations especially when their control cycles are different. For instance, the LEO satellite with a limited computing capability must consume much energy to infer the variations of the UE before making its decisions.

Recently, the deep reinforcement learning (DRL) technique has been regarded as a promising paradigm for handling sequential decision-making tasks with the merits of inferring environmental knowledge from the continuous interactions with the environment. Therefore, through regarding the others as a part of the environment, the LEO satellite and UE can act as distributed DRL agents to derive the optimal policies, leading to a multi-time-scale multi-agent DRL scheme. In the literature, there is still a lack of analytical foundation to achieve continuous policy improvements in such a DRL scheme, and particularly how to achieve the convergence of different agents' policies still remains a thorny issue \cite{Feriani2021}. To thoroughly address these key issues and pave these analytical foundations, in this paper, we focus on the two-time-scale two-agent scenario in NTNs (i.e., one LEO satellite and one UE are considered), which creates the research roadmap toward the multi-time-scale multi-agent DRL scheme.

In this paper, we propose a two-time-scale collaborative DRL scheme for the two-time-scale two-agent scenario, and apply this scheme to optimize the beam direction and resource allocation of the earth-fixed cell of NTNs. Specifically, the LEO satellite determines its transmitting beam and resource reservation scheme with a large-time-scale control cycle, and the UE optimizes its receiving beam pattern and resource access policy based on its data rate demands with a small-time-scale control cycle. To the best of our knowledge, this scheme is the pioneer to adopt two-time-scale two-agent DRL to solve resource optimizations of the earth-fixed cell of NTNs. Most notably, this paper provides a comprehensive convergence analysis for the investigated two-time-scale two-agent DRL, which can be adopted to other usage scenarios and extended to the general multi-agent scenarios. The main contributions and principles of this scheme include the followings.
\begin{itemize}
\item{With the objective of reducing the optimization problem size, we first propose a two-time-scale two-agent Markov decision process (TTMDP) model considering the mutual-impact between agents with different control cycles. The joint optimization of beam management and resource allocation is designed and formulated as large-time-scale and small-time-scale MDPs for the LEO satellite and UE, respectively. }

\item{Considering the limited computing capability of the LEO satellite, the UE with the sufficient computing capability handles most of computations for tackling the proposed TTMDP model. To this end, the trust region policy optimization (TRPO) is adopted to create monotonic policy improvement directions for both agents (in different tiers) at the UE side. Subsequently, an efficient finite-step rollout algorithm is developed for the LEO satellite to obtain the optimal policy along the improvement direction derived by the UE.}

\item{Since the performance convergence is the key to enable a distributed DRL scheme, we develop rigorous analytical foundations for the proposed TTMDP model, and derive the convergence guarantee of the proposed two-time-scale two-agent DRL scheme. Additionally, for each agent in the proposed scheme, the bounds on the convergence error and convergence time are also provided.}

\item{Extensive simulations are conducted, which show that the proposed scheme substantially improves the overall throughput over existing DRL-based schemes without proper collaboration designs. Additionally, compared to the combinations of searching-based/geometry-based beam optimization schemes and greedy-based/fixed/bandit-based resource allocation schemes, the proposed scheme can also tackle the trade-off between the overall throughput, resource utilization and computational complexity.}

\end{itemize}

The rest of this paper is organized as follows. The related work is presented in Section II. In Section III, the system model and problem formulation are specified. In Section IV, we provide the facilitation details and convergence analysis of the proposed two-time-scale collaborative DRL scheme. Furthermore, the simulation results are presented in Section V, and the paper is concluded in Section VI.

\section{Related Work}

As steerable beams are of most importance in the earth-fixed cell, many researches concentrate on the beam optimization of the LEO satellite. In \cite{Palacios2021}, a family of transmitting beam codebooks was designed through capturing the relative movements of the LEO satellites to UEs. In \cite{Roper2022}, a cooperative multi-satellite transmission scenario was considered, and a distributed precoding scheme for the satellite swarm was proposed to optimize the achievable rate at the ground gateway. Furthermore, the link establishment issue for LEO satellites with beam steering was tackled through a greedy matching algorithm in \cite{Mayorga2021}. While in \cite{Li2022}, the downlink channel model for the LEO-UE pair with uniform planar arrays (UPAs) was derived, and the optimal transmission strategy was also developed to maximize the ergodic sum rate of multiple UEs. Furthermore, to achieve a good trade-off between the performance and complexity, a beam-updating-frequency-reducing scheme was investigated in \cite{Zhao2021} to design long-term effective codebooks based on the slow variations of space-terrestrial links to replace the real-time adjustments. Nevertheless, in those existing works, it is assumed that a line-of-sight (LOS) link always exists between an LEO satellite and any UE, and thus only the slow variations of space-terrestrial links are considered (namely, relative positions and average channel gains) to optimize the statistical performance. However, in the practical NTN deployment, the LOS link may not generally exist. If the elevation angle of the LEO satellite moving path is not sufficiently large, the link between an LEO satellite and a ground UE may be scattered by the landform. In this case, the fast-scale variations (namely, multi-path effects and doppler shifts) should also be adequately taken into considerations to further enhance the instantaneous performance.

With the merit of developing data-driven schemes, DRL utilizing deep neural networks (DNNs) as powerful approximation methods has been introduced to enhance the adaptability of the resource management schemes to the environmental variations both in slow and fast scales. For example, DRL was adopted to optimize the bandwidth allocation for the multi-beam LEO satellite with dynamic traffic loads \cite{Liao2020}. Furthermore, in \cite{Lin2022}, the joint optimization of the beam pattern and bandwidth allocation was solved by a multi-agent DRL scheme through adopting each beam of the satellite as an agent. However, even with a non-DRL method, the LEO satellite with a limited computing capability cannot satisfy the real-time RRM requirement of the earth-fixed cell, especially facing the high-dimensional optimization of beam direction and resource allocation. To successfully perform the decision-making task with low complexity, the LEO satellite needs to offload the DRL model training task to terrestrial stations with powerful computing capabilities. To this end, distributed model-training architecture such as federated learning can also be utilized to train a global DRL model to be employed at the satellite side \cite{Cao2021,Razmi2022}. Nevertheless, since the number of DNN parameters is mainly determined by the optimization dimension, a large amount of signaling overheads can be incurred by frequent DRL-model-parameter exchanges between the LEO satellite and terrestrial stations.

To practically implement the distributed control architecture with low signaling overheads, stations involved in the resource optimization should actively configure corresponding resources, and thus the multi-agent DRL scheme for stations with different characteristics and capabilities should be constructed. Recently, multi-time-scale DRL schemes have been widely investigated for stations with different control cycles \cite{Talwar2021,Tan2019,Qin2019,Han2023,Chai2022}. In \cite{Tan2019}, deep Q-network and particle swarm optimization algorithms were adopted to derive the large-time-scale server selection policy and small-time-scale resource allocation policy, respectively. Additionally, a multi-time-scale coordination model was developed for self-organizing networks, and the corresponding resource optimization was solved by a Q-learning algorithm in \cite{Qin2019}. While in \cite{Han2023}, a hybrid actor-critic algorithm was proposed to tackle a two-time-scale MDP formulated for LEO-assisted task offloading. However, the most crucial and fundamental issue in existing multi-time-scale and multi-agent DRL methods is the lack of analytical foundation to guarantee the performance convergence. Although the convergence analysis of a two-time-scale DRL-based stochastic optimization scheme was provided in \cite{Chai2022}, the convergence of the multi-time-scale multi-agent DRL scheme remains a research hole. Different from prior works, this paper provides the policy improvement mechanism and corresponding convergence performance analysis for the two-time-scale two-agent DRL scheme, which offers a fundamental breakthrough toward multi-time-scale multi-agent DRL schemes.

\section{System Model and Problem Formulation}

\subsection{System Model}

\begin{figure*}
\centering
\includegraphics[scale = 0.48]{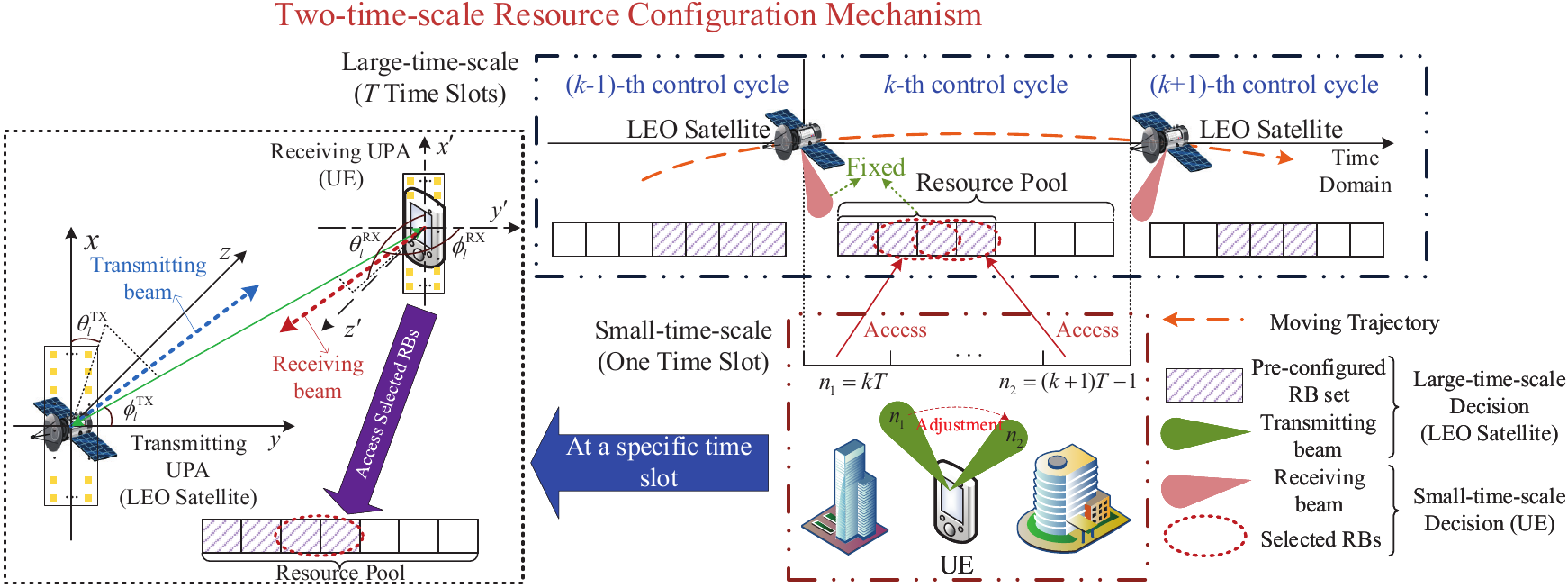}
\caption{The LEO downlink transmission model, in which one moving LEO satellite services the ground UE.}\label{fig:systemodel}
\end{figure*}

As illustrated in Fig. \ref{fig:systemodel}, in this paper, the downlink transmissions in the NTN are considered, in which a specific ground UE is serviced by an LEO satellite moving along a predesigned orbit. UPAs with different numbers of antennas are deployed both at the LEO satellite and ground UE, and consequently there are $N_t^x$ and $N_t^y$ antennas in $x$- and $y$-axis of LEO satellite’s UPA, respectively. The total number of antennas on an LEO satellite is therefore $N_t = N_t^x \times N_t^y $. At the UE side, there are  $N_r^{x'}$ and $N_r^{y'}$ antennas  in $x'$- and $y'$-axis of its UPA, respectively, and the total antenna number is $N_r = N_r^{x'} \times N_r^{y'}$.

\subsubsection{Channel Model}
According to \cite{Li2022}, the (small-scale) downlink channel gain between the LEO satellite and UE at time instant $t$ and frequency $f$ can be given as,
{\begin{equation}\label{eq:channel}
{\bm{H}}_{t,f}  = \sum\limits_{l=0}^{L-1}\alpha_l e^{j2\pi[tv_l-f\tau_l]}{\bm{a}}_r(\theta_l^{\text{Rx}}, \phi_l^{\text{Rx}}){\bm{a}}_t^{\text{H}}(\theta_l^{\text{Tx}}, \phi_l^{\text{Tx}}),
\end{equation}}
where $L$ is the number of multi-paths, $\alpha_l$, $v_l$ and $\tau_l$ are the complex-valued gain, Doppler shift and propagation delay at path $l$, respectively, and $(\cdot)^{\text{H}}$ is the conjugate transpose operation. Additionally, the transmitting 3D-steering vector at path $l$ with angles-of-departures (AoDs) of azimuth angle $\theta_l^{\text{Tx}}$ and elevation angle $\phi_l^{\text{Tx}}$ is defined as
{\begin{equation}
{\bm{a}}_t(\theta_l^{\text{Tx}}, \phi_l^{\text{Tx}}) = {\bm{a}}_{t,x}(\theta_l^{\text{Tx}}, \phi_l^{\text{Tx}}) \otimes {\bm{a}}_{t,y}(\theta_l^{\text{Tx}}, \phi_l^{\text{Tx}}),
\end{equation}}
where ${\bm{a}}_{t,x}(\theta_l^{\text{Tx}}, \phi_l^{\text{Tx}})$ and ${\bm{a}}_{t,y}(\theta_l^{\text{Tx}}, \phi_l^{\text{Tx}})$ are the steering vectors on the $x$- and $y$-directions, which can be given by
{\begin{align}\nonumber
{\bm{a}}_{t,x}(\theta_l^{\text{Tx}}, \phi_l^{\text{Tx}}) &= \frac{1}{\sqrt{N_{t}^{x}}}[1, e^{j\frac{2\pi}{\lambda}d_t\sin(\phi_l^{\text{Tx}})\cos(\theta_l^{\text{Tx}})}, \ldots,\\
&\quad  e^{j\frac{2\pi}{\lambda}(N_{t}^{x} - 1)d_t\sin(\phi_l^{\text{Tx}})\cos(\theta_l^{\text{Tx}})} ]^{\text{T}},\\\nonumber
{\bm{a}}_{t,y}(\theta_l^{\text{Tx}}, \phi_l^{\text{Tx}}) &= \frac{1}{\sqrt{N_{t}^{y}}}[1, e^{j\frac{2\pi}{\lambda}d_t\cos(\phi_l^{\text{Tx}})}, \ldots,\\
&\quad  e^{j\frac{2\pi}{\lambda}(N_{t}^{y} - 1)d_t\cos(\phi_l^{\text{Tx}})} ]^{\text{T}},
\end{align}}
where $\lambda$ is the wave length, $(\cdot)^{\text{T}}$ is the transpose operator and $d_t$ is the inter-antenna spacing of the transmitting UPA.

Similarly, the three-dimensional (3D) steering vector of receiving UPA at path $l$ with angles-of-arrivals (AoAs) of azimuth angle $\theta_l^{\text{Rx}}$ and elevation angle $\phi_l^{\text{Rx}}$ is given by ${\bm{a}}_r(\theta_l^{\text{Rx}}, \phi_l^{\text{Rx}}) = {\bm{a}}_{r,x'}(\theta_l^{\text{Rx}}, \phi_l^{\text{Rx}}) \otimes {\bm{a}}_{r,y'}(\theta_l^{\text{Rx}}, \phi_l^{\text{Rx}})$, where the steering vectors with inter-antenna spacing of $d_r$ on the ${x}^{\prime}$- and ${y}^{\prime}$-directions are expressed by
{\begin{align}
{\bm{a}}_{r,x'}(\theta_l^{\text{Rx}}, \phi_l^{\text{Rx}}) &= \frac{1}{\sqrt{N_{r}^{x'}}}[1, e^{j\frac{2\pi}{\lambda}d_r\sin(\phi_l^{\text{Rx}})\cos(\theta_l^{\text{Rx}})}, \ldots,\\\nonumber
&  e^{j\frac{2\pi}{\lambda}(N_{r}^{x'} - 1)d_r\sin(\phi_l^{\text{Rx}})\cos(\theta_l^{\text{Rx}})} ]^{\text{T}},\\\nonumber
{\bm{a}}_{r,y'}(\theta_l^{\text{Rx}}, \phi_l^{\text{Rx}}) &= \frac{1}{\sqrt{N_{r}^{y'}}}[1, e^{j\frac{2\pi}{\lambda}d_r\cos(\phi_l^{\text{Rx}})}, \ldots,\\
&  e^{j\frac{2\pi}{\lambda}(N_{r}^{y'} - 1)d_r\cos(\phi_l^{\text{Rx}})} ]^{\text{T}}.
\end{align}}

\subsubsection{Downlink Transmission Model}
The \emph{orthogonal frequency-division multiple access} (OFDMA) has been widely adopted by state-of-the-art NTNs such as the 3GPP NR-NTN. For OFDMA downlink transmissions, a radio resource pool is composed of $M$ time-frequency resource blocks (RBs) indexed by ${\cal{M}} = \{0, \ldots, M-1\}$. Particularly, multiple RBs can be allocated to the UE at each time slot, and a binary indicator $b_{n,m}$ is adopted to indicate whether RB $m$ is allocated to the UE at time slot $n$, i.e., $b_{n,m} = 1$ if RB $m$ is allocated to the UE at time slot $n$, and $b_{n,m} = 0$, otherwise. In this case, the downlink channel gain in \eqref{eq:channel} can be rewritten as ${\bm{H}}_{n,m}$ by replacing $t = nT_s$ and $f= \frac{m}{T_s}$, where $T_s$ is the time duration of one OFDM symbol.

Denote transmitting beam of angles $(\theta_t^n, \phi_t^n)$ at the LEO satellite and receiving beam of angles $(\theta_r^n, \phi_r^n)$ at the UE as ${\bm{w}}_t(\theta_t^n, \phi_t^n) \in \mathbb{C}^{N_t \times 1}$ and ${\bm{w}}_r(\theta_r^n, \phi_r^n) \in \mathbb{C}^{N_r \times 1}$, respectively, with the same definitions as ${\bm{a}}_t(\cdot)$ and ${\bm{a}}_r(\cdot)$. The received signal at the UE side in RB $m$ at time slot $n$ can be expressed based on the downlink channel gain defined in \eqref{eq:channel}, i.e.,
{\begin{align}\nonumber
&\quad {{y}}_{n, m}(\theta_t^n, \phi_t^n, \theta_r^n, \phi_r^n)\\\nonumber
& = \sqrt{P_tL_n}{{\bm{w}}_r}(\theta_r^n, \phi_r^n)^{\text{H}}{\bm{H}}_{n, m}{\bm{w}}_t(\theta_t^n, \phi_t^n){{x}}_{n, m}\\
&\quad + {{\bm{w}}_r}(\theta_r^n, \phi_r^n)^{\text{H}}{\bm{z}},
\end{align}}
where  $P_t$ is the transmitting power, $L_n$ is the pathloss between the LEO satellite and UE at time slot $n$, ${{x}}_{n, m}$ is the transmit signal with unit power, and $\bm{z} \sim {\cal{CN}}({\bm{0}}, \delta_z^2{\bm{I}^{N_r \times N_r}})$ is the additive white Gaussian noise (AWGN) with zero mean and variance of $\delta_z^2 = k_BT_sB$, where $k_B$, $T_s$ and $B$ are the Boltzmann constant, noise temperature, and bandwidth of each RB, respectively. Therefore, the signal-to-noise-ratio (SNR) of the UE in RB $m$ at time slot $n$ is given as
{\begin{align}\nonumber
&\quad \varpi_{n, m}(\theta_t^n, \phi_t^n, \theta_r^n, \phi_r^n)\\\label{eq:snr}
& = \frac{P_tL_n|{{\bm{w}}_r(\theta_r^n, \phi_r^n)}^{\text{H}}{\bm{H}}_{n, m}{\bm{w}}_t(\theta_t^n, \phi_t^n)|^2}{N_r\delta_z^2}.
\end{align}}
From \eqref{eq:snr}, it can be found that the SNR at the UE side is mainly determined by the beam gain $|{{\bm{w}}_r}^{\text{H}}{\bm{H}}_{n,m}{\bm{w}}_t|^2$. The receiving rate of the UE in RB $m$ at time slot $n$ therefore can be expressed as
{\begin{align}\nonumber
&\quad  c_{n, m}(\theta_t^n, \phi_t^n, \theta_r^n, \phi_r^n) \\
& = B\log_2(1 + \varpi_{n, m}(\theta_t^n, \phi_t^n, \theta_r^n, \phi_r^n)).
\end{align}}

\vspace{-1em}
\subsection{Two-Time-Scale Resource Configuration Mechanism}

In 3GPP NR-NTNs, the UE with a sufficient energy supply can act as a powerful computing platform, while the LEO satellite with a limited energy supply only has limited computing capability. In this case, as shown in Fig. \ref{fig:systemodel}, we propose a two-time-scale resource configuration mechanism, in which the LEO satellite and UE configure radio resources with different time-scales (or control cycles). Specifically, aiming at alleviating the computing burdens of the LEO satellite, the periodic resource configuration policy  \cite{Zhao2021} is adopted for the LEO satellite to configure its resource with a large control cycle. Then, the UE performs small-scale adjustments on corresponding resources to enhance its receiving rate performance.

For the LEO satellite, toward constructing the earth-fixed cell for the UE, the transmitting beam should be optimized along the moving trajectory. Moreover, since the number of available RBs is normally limited, the utilization of RBs is of crucial importance. If the LEO satellite allocates all the RBs to the UE within its large control cycle, severe wasting of RBs may happen. To improve the utilization of RBs, only a minimum number of RBs with respect to the best channel gains between the LEO satellite and UE should be reserved by the LEO satellite for the UE to form a pre-configured RB set, which is denoted as ${\cal{M}}_{\text{LEO}}^n$ at time slot $n$. Therefore, the LEO satellite should optimize the transmitting beam and the pre-configured RB set every $T$ time slots. This indicates that one single control cycle of the LEO satellite is composed of $T$ time slots, and the control cycle index $k$ can be given by $k = \lfloor{\frac{n}{T}}\rfloor$. Therefore, ${\cal{M}}_{\text{LEO}}^n \subseteq   {\cal{{M}}}$ can also be denoted as ${\cal{{M}}}_{\text{LEO}}^k$. Particularly, the control cycle length $T$ can be determined based on the ephemeris of the LEO satellite. Following the transmitting beam direction and pre-configured RB set of the LEO satellite, the UE should adjust its receiving beam direction and access RBs selected from ${\cal{{M}}}_{\text{LEO}}^k$ to satisfy its receiving rate demand $D_{\text{UE}}^n$ at time slot $n$. Thus, the control cycle of the UE is one time slot{\small{\footnote{Generally, the control cycle of the UE is defined based on the transmission duration of the satellite link, and can be with an arbitrary length in the time domain. In this paper, one time slot is adopted as an example to indicate the real-time RRM optimization.}}}, and its control cycle index is equivalent to the time slot index $n$, i.e., $i = n$. Hence, its receiving rate demand $D_{\text{UE}}^n$ can also be denoted as $D_{\text{UE}}^i$. Please also note that the time slot lengths of the LEO satellite and UE are the same and their time slot boundaries are aligned, and the time synchronization is out of the scope of this paper.

\subsection{Problem Formulation}

Based on the above two-time-scale resource configuration mechanism, an RB minimization problem is mathematically formulated.

\begin{myOpt}
A long-term minimization problem for the number of utilized RBs with respect to transmitting-receiving beam optimization and RB selection from time slot $0$ to time slot $N - 1$ is given by
{\setlength{\abovedisplayskip}{3pt}
\setlength{\belowdisplayskip}{3pt}
\begin{align}
& \mathop{\min}\limits_{\{\theta_t^n, \phi_t^n, \theta_r^n, \phi_r^n\}, \{b_{n, m}\}, \{{\cal{{M}}}_{\text{LEO}}^n\} } \frac{1}{N}\sum\limits_{n = 0}^{N-1}\sum\limits_{m=0}^{M - 1}b_{n,m}\\\label{eq:cons_rate}
\mathrm {s.t.} & \sum\limits_{m \in {\cal{{M}}}_{\text{LEO}}^n}b_{n, m}c_{n, m}(\theta_t^n, \phi_t^n, \theta_r^n, \phi_r^n) \ge D_{\text{UE}}^n, \forall n,\\
& b_{n,m} \in \{0, 1\}, \forall m \in {\cal{{M}}}, \forall n,\\
& 0 \le \zeta \le \pi, \forall \zeta \in \{\theta_t^n, \phi_t^n, \theta_r^n, \phi_r^n\}, \forall n,\\
& {\cal{{M}}}_{\text{LEO}}^n \subseteq   {\cal{{M}}}, \forall n.
\end{align}}
\end{myOpt}

In \textbf{Optimization 1}, the variables belonging to heterogeneous agents (namely, LEO satellite and UE) are of different characteristics (e.g., valid during multiple time slots or at each time slot) and coupled with each other. Consequently, if a centralized approach to solve this optimization is applied, the joint optimization composed of these variables with different time scales results in an extremely high dimensional decision space, and a huge number of decision trajectories over multiple time slots need to be explored repeatedly to estimate the performances under all the variable combinations, leading to an unaffordable computational overheads.

To address the above issue, the multi-agent approach becomes a promising remedy. Specifically, these optimization variables can be separated into distinct sets, and a series of single-stage optimizations are defined by different sets of variables. To this end, according to whether a variable should be determined by the LEO satellite or UE, \textbf{Optimization 1} can be transformed into multiple distinct but dependent MDPs. Regard the LEO satellite and UE as the high-tier agent (with a longer control cycle) and low-tier agent (with a shorter control cycle), respectively, and their corresponding MDPs can be denoted as $<{\cal{S}}_H, {\cal{A}}_H, {\cal{P}}_H(\pi_H, \pi_L), {R}_H(\pi_H, \pi_L)>$ and $<{\cal{S}}_L, {\cal{A}}_L, {\cal{P}}_L(\pi_H, \pi_L), {R}_L(\pi_H, \pi_L)>$, where ${\cal{S}}_H$, ${\cal{A}}_H$, $\pi_H$, ${\cal{P}}_H(\pi_H, \pi_L)$ and ${R}_H(\pi_H, \pi_L)$ (${\cal{S}}_L$, ${\cal{A}}_L$, $\pi_L$, ${\cal{P}}_L(\pi_H, \pi_L)$ and ${R}_L(\pi_H, \pi_L)$) are state space, action space, policy function, i.e., ${\bm{a}}_H = \pi_H({\bm{s}}_H)$ (${\bm{a}}_L = \pi_L({\bm{s}}_L)$), state transition probability and reward function of the higher-tier (lower-tier) agent, respectively. Since the state transition probabilities and reward functions of agents are influenced by the policies of each other, the value functions can be given as
{
\begin{align}\nonumber
V^{\pi_H}({\bm{s}}_H;\pi_L) &= \mathbb{E}_{\pi_H}[\sum_{k=0}^{\infty}{\gamma_H}^{k}R_H({\bm{s}}_H^k, \pi_H({\bm{s}}_H^k), \\
& \quad \tau_L(\pi_L))|{\bm{s}}_H^0 = {\bm{s}}_H],\\\nonumber
V^{\pi_L}({\bm{s}}_L;\pi_H) &= \mathbb{E}_{\pi_L}[\sum_{i=0}^{\infty}{\gamma_L}^{i}R_L({\bm{s}}_L^i, \pi_L({\bm{s}}_L^i), \\
& \quad \tau_H(\pi_H))|{\bm{s}}_L^0 = {\bm{s}}_L],
\end{align}}
where $\mathbb{E}[\cdot]$ is the expectation operator, and $\gamma_H$ and $\tau_H(\pi_H)$ ($\gamma_L$ and $\tau_L(\pi_L)$) are the discount factor and decision trajectory of the higher-tier (lower-tier) agent, respectively. Thus, the combination of the above two MDPs of the LEO satellite and UE leads to a TTMDP model.

Due to the difference of $T$ time slots between control cycles of the LEO satellite and UE, the optimization of deriving the optimum policies of the TTMDP can be formulated as the sum of accumulated rewards of agents at different tiers, i.e.,
{
\begin{align}\nonumber
& \mathop{\max}\limits_{\pi_H, \pi_L} \mathbb{E}[\sum\limits_{k=0}^{N_H}[R_H({\bm{s}}_H^k, \pi_H({\bm{s}}_H^k), \tau_L(\pi_L))\\
&\quad + \sum\limits_{p=0}^{T-1}R_L({\bm{s}}_L^{Tk+p}, \pi_L({\bm{s}}_L^{Tk+p}), \tau_H(\pi_H))]],
\end{align}}
where $N_H = \lfloor{\frac{N-1}{T}}\rfloor$. To tackle different agents’ MDPs, the concept of the leader-follower game \cite{Lien2016} architecture can be adopted when accurate transition probabilities between states of these MDPs are available at the leader side, which is used for the leader to calculate the best decisions not only for itself but also for the follower (also known as the backward induction). How to solve the TTMDP with backward induction will be detailed in the following section.

\section{Two-time-scale Collaborative DRL Scheme in NTN}

Through transforming \textbf{Optimization 1} as the TTMDP, the LEO satellite and UE as higher- and lower-tier agents perform decision-making tasks with different control cycles to collaboratively solve corresponding MDPs, leading to a two-time-scale feature. Please note that the basic time-scale is not fixed in the time domain, and it refers to the time duration when one decision-making task is completed by the agent, namely, the control cycle. In this case, there are two main obstacles: 1) the design of collaboration operations for heterogeneous agents; 2) the convergence guarantee of the proposed scheme. In 3GPP NR-NTNs, with the rapid development of the moving computing platform and the sufficient energy supply, the UE (such as a hand-held device or vehicle) with the sufficient energy supply has superior computing capabilities over the LEO satellite (whose computing platforms are affected negatively by the space environment), and it can therefore support manifold computing tasks in DRL-model-updating and policy-prediction. In this case, motivated by the backward induction, the UE with the powerful computing capability should improve both the value functions of the LEO satellite and UE simultaneously (similar to a leader in the leader-follower game). Therefore, owing to the coupled policies of the LEO satellite and UE, we should first derive the effects of their decision-makings on respective value functions, which are then adopted by the UE to calculate the improvement directions of both the value functions at the UE side and the LEO satellite side. Then, the LEO satellite only needs to perform a finite-step rollout algorithm to determine its best policies based on the obtained reference decision trajectory from the UE. Through updating policies at the LEO satellite and UE iteratively, both the value functions can converge, and we will present corresponding analysis.

\subsection{MDP Formulations of the LEO satellite and UE}

\begin{figure*}
\centering
\includegraphics[scale = 0.33]{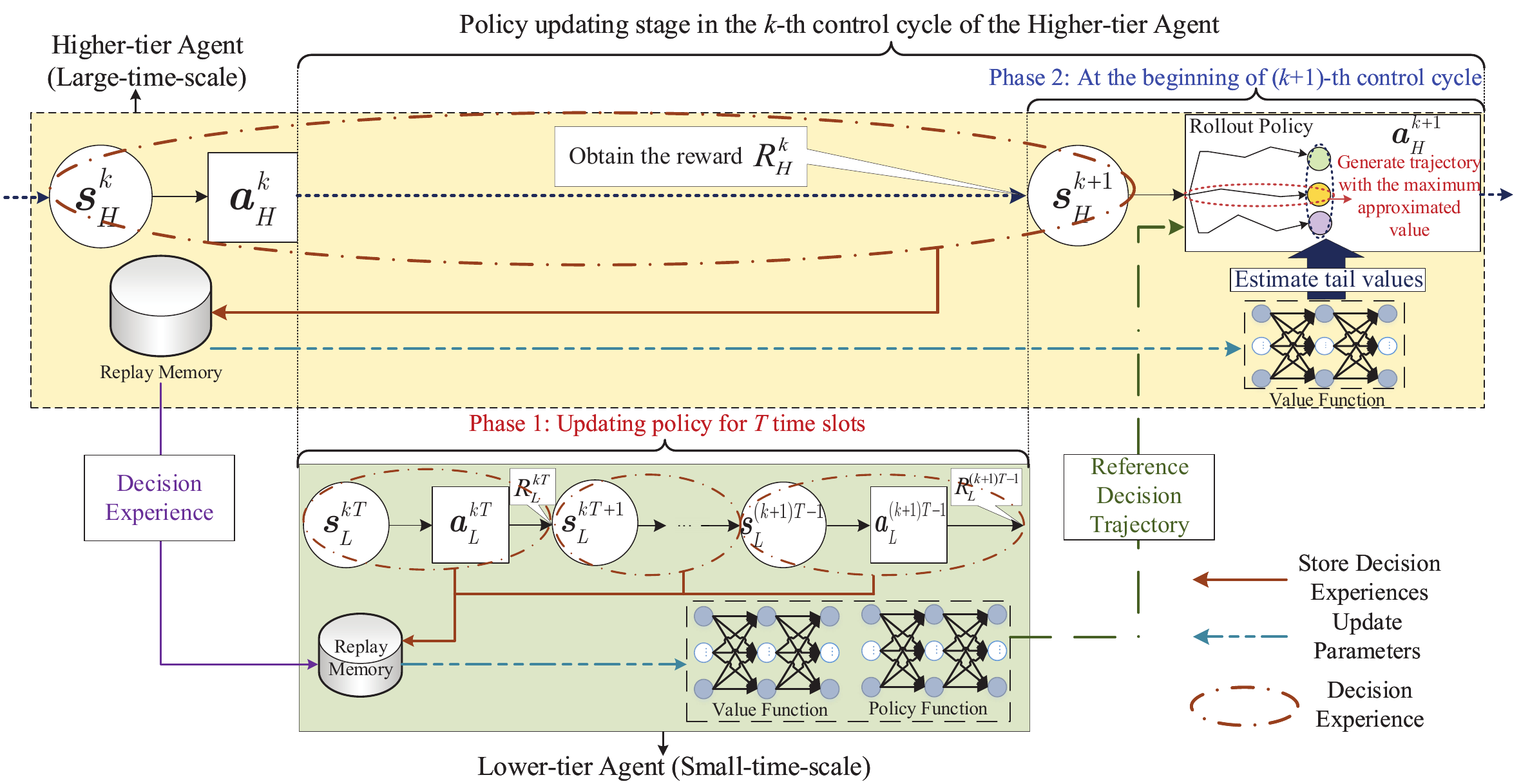}
\caption{The proposed two-time-scale collaborative DRL scheme in this paper.}\label{fig:mmdp}
\end{figure*}

As shown in Fig. \ref{fig:mmdp}, the two-time-scale collaborative DRL scheme constituted by the LEO satellite and UE with different control cycles is proposed. Since the performance of the DRL-based algorithm is determined by the design of state space, action space and reward function of the MDP, we should first present the state spaces, action spaces and reward functions for respective MDPs of different agents in the TTMDP model according to \textbf{Optimization 1}.

\begin{myDef} MDP of the LEO satellite

\begin{itemize}
\item{State Space:} Since the LEO satellite needs to capture the variation patterns on relative position and receiving data demand of the UE within its control cycle, a state of the LEO satellite is composed of two groups: 1) position of the LEO satellite, i.e., $p_{\text{LEO}}^k = (x_{\text{LEO}}^k, y_{\text{LEO}}^k, z_{\text{LEO}}^k)$; 2) average SNR over the selected RBs at each time slot within its $(k-1)$-th control cycle, i.e.,  $\bar{\varpi}_{k-1} = \frac{1}{\sum\limits_{m \in {{\cal{{M}}}}}{b^L_{k-1, m}}}\sum\limits_{m \in {{\cal{{M}}}}} b^L_{k-1, m}\varpi_{(k-1)T+p, m}, p = 0, \ldots, T-1$. Thus, the state can be given as
\begin{equation}
{\bm{s}}_H^k = \{p_{\text{LEO}}^k; \bar{\varpi}_{(k-1)T}, \ldots, \bar{\varpi}_{(k-1)T + (T -1)}\}.
\end{equation}
where $b_{k,m}^{L} = 1$ if RB $m$ is selected by the LEO satellite within $k$-th control cycle, and $b_{k,m}^{L} = 0$, otherwise.

\item{Action Space:} The LEO satellite needs to determine its transmitting beam direction and select a set of RB candidates for each control cycle, i.e.,
\begin{equation}
{\bm{a}}_H^k =\{\theta_t^k, \phi_t^k; {\cal{{M}}}_{\text{LEO}}^k\}.
\end{equation}

To decrease the dimension of action space, a non-codebook design with the discrete adjustment angle is adopted to determine the transmitting beam direction{\small\footnote{If continuous beam directions are considered, there is a continuous-discrete hybrid action space. In this case, multiple DNNs outputting continuous/discrete actions should be introduced, and the integration of those DNNs should be carefully designed, which is similar to the multi-pass deep Q-network\cite{Bester2019}. Therefore, the solution for such a hybrid action space will be investigated in the future research.}}. Namely, an adjustment angle is selected from the LEO satellite's discrete angular set (e.g., $\{-\Delta, 0, \Delta \}$, where $\Delta$ is the unit adjustment angle), which is added to the initial beam direction as a new beam direction within its current control cycle. Particularly, ${\cal{{M}}}_{\text{LEO}}^k = \{b_{k,m}^{L}, m \in {\cal{M}}\}$.

\item{Reward Function:} According to constraint \eqref{eq:cons_rate} in \textbf{Optimization 1}, the LEO satellite should enhance the SNR in \eqref{eq:snr}, and then select RB sets with respect to the highest average receiving rate so that the number of RBs utilized by the UE can be reduced. Therefore, the reward function is defined as the average receiving rate within time slots satisfying constraint \eqref{eq:cons_rate} within each control cycle, i.e.,
{\begin{align}\label{eq:h_reward}\nonumber
&\quad R_H^k  = \frac{1}{T}\sum\limits_{p = 0}^{T - 1}\mathbb{I}_{\text{Demand}}\cdot\{\\
& \quad \frac{\sum\limits_{m \in {\cal{M}}}b^L_{k,m}b_{kT + p, m}c_{kT + p, m}(\theta_t^k, \phi_t^k, \theta_r^{kT + p}, \phi_r^{kT + p})}{\sum\limits_{m \in {\cal{{M}}}}b^L_{k,m}b_{kT + p, m}}\},
\end{align}}
where $\mathbb{I}_{\text{Demand}}$ is the indicator function of constraint \eqref{eq:cons_rate}, i.e., $\mathbb{I}_{\text{Demand}} = 1$ if \eqref{eq:cons_rate} is satisfied, and $\mathbb{I}_{\text{Demand}} = 0$, otherwise.
\end{itemize}
\end{myDef}

\begin{myDef} MDP of the UE

\begin{itemize}
\item{State Space:}To capture the channel variations, there are two groups in a state of UE: 1) SNR in each RB at the last time slot, i.e., $\{b_{k,m}^{L}b_{i-1, m}\varpi_{i-1, m}, \forall m \in {\cal{{M}}}\}$; 2) average signal strengths at receiving antennas over selected RBs at the last time slot, i.e., ${\bm{\Gamma}}_{i-1} = \frac{1}{\sum_{m \in {\cal{M}}}b_{k,m}^{L}b_{i-1, m}}\sum\limits_{m \in \cal{M}}b_{k,m}^{L}b_{i-1, m}|{\bm{H}}_{i-1, m}{\bm{w}}_t(\theta_t^k, \phi_t^k)|$. Thus, the state can be given as
{\begin{align}\nonumber
{\bm{s}}_L^i & = \{b_{k,0}^{L}b_{i-1, 0}\varpi_{i-1, 0}, \ldots,\\
& \quad b_{k,M-1}^{L}b_{i-1, M-1}\varpi_{i-1, M-1}; {\bm{\Gamma}}_{i-1}\}.
\end{align}}

\item{Action Space:} The UE at each time slot needs to adjust its receiving beam direction and access proper RBs from pre-planned RB candidates, i.e.,
\begin{equation}
{\bm{a}}_L^i = \{\theta_r^i, \phi_r^i; b_{k,0}^{L}b_{i,0}, \ldots, b_{k,M-1}^{L}b_{i, M-1}\}.
\end{equation}

The UE also determines its receiving beam direction through adding an adjustment angle sampled from its discrete angular set, e.g., $\{-3\Delta, -2\Delta, \ldots , 3\Delta \}$.

\item{Reward Function:} To minimize the number of utilized RBs, the RB with respect to the highest receiving rate should be selected first to satisfy the constraint \eqref{eq:cons_rate}. To this end, the instantaneous reward function is composed of the average receiving rate of selected RBs and a satisfactory punishment to avoid wasting RBs, i.e.,
     \begin{equation}\label{eq:l_reward}
{\hat{R}}_L^i = \frac{\sum\limits_{m \in {\cal{M}}}b_{k,m}^{L}b_{i,m}c_{i,m}(\theta_t^k, \phi_t^k, \theta_r^i, \phi_r^i)}{\sum\limits_{m \in {\cal{M}}}b_{k,m}^{L}b_{i,m}} + \eta \Omega_{i},
\end{equation}
where $\Omega_{i} = \min[{\sum\limits_{m \in {\cal{M}}}b_{k,m}^{L}b_{i,m}c_{i,m} - D_{\text{UE}}^{i}}, 0]$ is the satisfactory punishment, and $\eta$ is a punishment coefficient. Moreover, to alleviate fast variations in each time slot, a first-in-first-out (FIFO) buffer ${\cal{B}}_r$ is adopted to perform moving average over instantaneous rewards, and then the average value of ${\cal{B}}_r$ is adopted as the final reward function of the UE, i.e., $R_L^i = \mathbb{E}_{{\cal{B}}_r}[{\hat{R}}_L^i]$, where $\mathbb{E}_{{\cal{B}}_r}[\cdot]$ is the expectation over instantaneous rewards in buffer ${\cal{B}}_r$.
\end{itemize}
\end{myDef}

\subsection{Design for Two-time-scale Collaborative DRL Scheme}

Before introducing details of the proposed two-time-scale scheme, we should first overview the conventional schemes for addressing ``multi-time-scale MDP'' issues. In \cite{Chang2003}, although the state and action space for each agent at different time scales are non-overlapping, the state transition probabilities of the lower-tier agent is mainly determined by the large-time-scale policy. Moreover, the proposed solution in \cite{Chang2003} is only suitable for MDP with a low-dimensional decision space. However, in our MDP formulation, the state transition probabilities of agents at different time scales are influenced by each other. In this case, our considered MDP model is also different from other works on MDPs with ``goals''\cite{Kulkarni2016}, ``options''\cite{Sutton1999} or ``skills''\cite{Li2019}, in which only either state spaces or action spaces at different time scales are distinct and the state transition probabilities of the MDP are not influenced by these hierarchical variables. To formulate the mutual influences between agents at different time scales, the game-theoretic view \cite{Wen2021} has been widely adopted in recent works, and the multi-agent system is normally formulated as a differential games such as differential Stackelberg game \cite{Yang2021}, in which strong assumptions of knowing variations of the environment and the existence of the unique best response at a smaller time scale need to be imposed.

With the merit of addressing non-stationarity in multi-agent systems, the sequential updating scheme \cite{Bertsekas2021} is adopted for agents updating their policies independently at different time scales through observing other agent's decision trajectory $\tau_L(\pi_L)$ or $\tau_H(\pi_H)$. This indicates that our proposed scheme is synchronous. Additionally, with the spirit of backward induction, the DRL scheme to address the TTMDP should be designed based on the practical computing capabilities of the LEO satellite and UE. Considering the sufficient computing capability, the UE as the lower-tier agent needs to calculate the policy improvement directions for both agents in different tiers, and then the LEO satellite as the higher-tier agent updates its policy based on the reference improvement direction provided by the UE. Therefore, as shown in Fig. \ref{fig:mmdp}, one policy updating stage in our proposed scheme is composed of two phases. In the first phase, the lower-tier agent needs to maximize its value function over its original policy $\pi'_L$ given a fixed higher-tier policy. Moreover, since the higher-tier MDP is influenced by the lower-tier policy, the lower-tier agent also needs to reconsider its impact on the higher-tier agent and provides a ``correct'' ascent direction for improving the value function of the higher-tier agent over its policy $\pi_H$. To this end, the policy of the lower-tier agent should be updated along an average direction for improving advantages of different agents simultaneously. If we regard the future higher-tier policy $\bar{\pi}_H$  as a function of $\pi_L$, i.e., $\bar{\pi}_H = f(\pi_L)$, the policy updating rule at the lower-tier agent can be
{\begin{align}\label{eq:sum_value}\nonumber
\bar{\pi}_L &= \mathop{\arg\max}\limits_{\pi_L \in {\cal{A}}_L}\{(\rho_L(\pi_L, f(\pi_L))-\rho_L(\pi'_L, \pi_H)) \\
&\quad + (\rho_H(\pi_H, \pi_L)-\rho_H(\pi_H, \pi'_L))\},
\end{align}}
where $\rho_L(\pi'_L, \pi_H) = \sum_{{\bm{s}}_L}d_L^{\pi'_L, \pi_H}(s_H)\sum_{{\bm{a}}_L \in {\cal{A}}_L}\pi_L({\bm{a}}_L|{\bm{s}}_L)\\Q_{\pi'_L}({\bm{s}}_L, {\bm{a}}_L;\pi_H)$ with the state distribution of $d_L^{\pi'_L, \pi_H}({\bm{s}}_L)$ is the expected value function of the lower-tier agent under policies of $\pi'_L$ and $\pi_H$, in which $Q_{\pi'_L}({\bm{s}}_L, {\bm{a}}_L;\pi_H) = \mathop{\mathbb{E}}_{{\bm{s}}'_H \sim P_L({\bm{s}}_H, \pi'_L, \pi_H)}[R_H({\bm{s}}_H, {\bm{a}}_H, \tau_L(\pi'_L))+{\gamma_H}V^{\pi_H}({\bm{s}}'_H;\pi'_L)] \\ = A_{\pi'_L}({\bm{s}}_L, {\bm{a}}_L; \tau_H(\pi_H)) + V^{\pi_H}({\bm{s}}_H;\pi'_L)$ is the state-action value function, while $\rho_H(\pi_H, \pi'_L)$ is the expected value function of the higher-tier agent under policies of $\pi'_L$ and $\pi_H$, and can be defined in a similar way to the state distribution of $d_H^{\pi'_L, \pi_H}(s_H)$.

To obtain the desired policy improvement direction for the lower-tier agent, we need to analyze the impact of the lower-tier agent's policy to the higher-tier agent. Based on the independent decision-making assumption in \cite{Wen2021} and trust region policy assumption in \cite{Schulman2015}, the policy improvement bounds incurred by the lower-tier agent can be derived, which are shown in the following \textbf{Lemma 1}.

\begin{myLem}
When the lower-tier agent updates its policy from $\pi'_L$ to $\pi_L$ with a policy-updating coefficient $\alpha_L$, the policy improvement bounds for the higher-tier can be given as
\begin{align}\nonumber
& \quad \rho_H(\pi_H, \pi_L) - \rho_H(\pi_H, \pi'_L) \\ \label{eq:pi_htier}
& \ge L_{\pi_H}(\pi_L) - \frac{4\epsilon_H(\pi_L)(1 -(1 -\alpha_L)^T)^2\gamma_H}{(1-\gamma_H)(1-(1-\alpha_L)^T\gamma_H)},
\end{align}
where $\epsilon_H(\pi_L) = \max_{{\bm{s}}_H, {\bm{a}}_H \sim \pi_H,  \tau'_L \sim \pi'_L}|A_{\pi_H}({\bm{s}}_H, {\bm{a}}_H; \tau_L(\pi'_L))|$ is the maximum value of the higher-tier agent's original advantages. $L_{\pi_H}(\pi_L)$ is the higher-tier agent's expected advantage function from policies of $(\pi'_L, \pi_H)$ to policies of $(\pi_L, \pi_H)$, which can be given as
{\begin{align}\nonumber
&\quad L_{\pi_H}(\pi_L) \\
&= \sum\limits_{{\bm{s}}_H}d_H^{\pi'_L, \pi_H}(s_H)\sum\limits_{{\bm{a}}_H \in {\cal{A}}_H}\pi_H({\bm{a}}_H|{\bm{s}}_H)A_{\pi_L}({\bm{s}}_H, {\bm{a}}_H; \tau_L(\pi_L)).
\end{align}}
On the other hand, the policy improvement bounds for the lower-tier agent
\begin{align}\nonumber
&\quad \rho_L(\pi_L, f(\pi_L)) - \rho_L(\pi'_L, \pi_H) \\\nonumber
& \ge L_{\pi'_L, \pi_H}(\pi_L, f(\pi_L))\\\nonumber
&\quad - 4(1 - (1-\alpha_L)(1 - \alpha_H)^{\frac{1}{T}}) \cdot \epsilon_L(\pi_H) \cdot \{\frac{1}{1-\gamma_L} \\\label{eq:pi_ltier}
&\quad - \frac{1-{\gamma_L}^T(1-\alpha_L)^T}{(1 - \gamma_L(1-\alpha_L))(1 - {\gamma_L}^T(1-\alpha_L)^T(1- \alpha_H))}\},
\end{align}
where $\alpha_H$ is the higher agent's policy-updating coefficient, and $\epsilon_L(\pi_H) = \max_{{\bm{s}}_L, {\bm{a}}_L \sim \pi'_L, \tau_H \sim \pi_H}|A_{\pi'_L}({\bm{s}}_L, {\bm{a}}_L; \tau_H(\pi_H))|$ is the maximum value of the lower-tier agent's original advantages. $L_{\pi'_L, \pi_H}(\pi_L, f(\pi_L))$ is the lower-tier agent's expected advantage function from policies of $(\pi'_L, \pi_H)$ to policies of $(\pi_L, f(\pi_L))$, which can be given as
{\begin{align}\nonumber
&\quad L_{\pi'_L, \pi_H}(\pi_L, f(\pi_L))\\
& = \sum\limits_{{\bm{s}}_L}d_L^{\pi'_L, \pi_H}({\bm{s}}_L)\sum\limits_{{\bm{a}}_L \in {\cal{A}}_L}\pi_L({\bm{a}}_L|{\bm{s}}_L) A_{\pi_L}({\bm{s}}_L, {\bm{a}}_L; \tau_H(f(\pi_L))).
\end{align}}
\end{myLem}

\begin{proof}
Please refer to Appendix A.
\end{proof}

With the facilitation of \textbf{Lemma 1}, a proper policy can be obtained through continuously improving the derived lower bounds. Considering that the policy of the higher-tier agent is near-stationary (i.e., $f(\pi_L) \approx \pi_H$), maximizing the above policy bounds can be equivalent to maximizing $L_{\pi_H}(\pi_L)$ and $L_{\pi'_L, \pi_H}(\pi_L, \pi_H)$ under the lower-tier agent policy $\pi'_L$  with the trust region constraints of ${\cal{O}}({\alpha_L}^T)$ and ${\cal{O}}({\alpha_H}^{\frac{1}{T}-1})$. To tackle the optimization with given policy constraints, we adopt the TRPO algorithm \cite{Schulman2015} in DRL to update the policy of the lower-tier agent. Trough adopting the second-order Taylor expansion method \cite{Schulman2015} to tackle the trust region constraints, the cumulative reward performance of the TRPO algorithm is similar to that of the commonly adopted proximal policy optimization algorithm. Specifically, at the lower-tier agent, two DNNs are deployed to approximate the policy and value function, respectively. Based on the policy updating rule in \eqref{eq:sum_value}, the loss function for updating the policy function parameters $\bm{\varphi}_L$ is evaluated based on the product of the policy ratio $\frac{\pi_L({\bm{a}}_L^{n}|{\bm{s}}_L^{n}; \bm{\varphi}_L)}{\pi_L({\bm{a}}_L^{n}|{\bm{s}}_L^{n}; \bm{\varphi}^{'}_L)}$ and the sum of all the agents' advantages (i.e., $\tilde{A}_L({\bm{s}}_L^{n}, {\bm{a}}_L^{n})$ and $\tilde{A}_H({\bm{s}}_H^{n}, {\bm{a}}_H^{n})$),
\begin{equation}\label{eq:policy_loss}
{\cal{L}}_{p} = \mathbb{E}[\frac{\pi_L({\bm{a}}_L^{n}|{\bm{s}}_L^{n}; \bm{\varphi}_L)}{\pi_L({\bm{a}}_L^{n}|{\bm{s}}_L^{n}; \bm{\varphi}^{'}_L)}(\tilde{A}_L({\bm{s}}_L^{n}, {\bm{a}}_L^{n})+\tilde{A}_H({\bm{s}}_H^{n}, {\bm{a}}_H^{n}))],
\end{equation}
where the expectation operator $\mathbb{E}[\cdot]$ is executed over the sampled state-action pair $({\bm{s}}_L^{n}, {\bm{a}}_L^{n})$. $\tilde{A}_L({\bm{s}}_L^{n}, {\bm{a}}_L^{n})$ and $\tilde{A}_H({\bm{s}}_H^{n}, {\bm{a}}_H^{n})$ are the estimated advantages for higher-tier and lower-tier agents, respectively, and $\bm{\varphi}^{'}_L$ is the original policy function parameters. Additionally, the value function parameters ${\bm{\theta}}_L$ is updated through minimizing the mean square error (MSE) between the estimated value and practical cumulative reward at each state, i.e.,
\begin{equation}\label{eq:critic_loss}
{\cal{L}}_{v} = \mathbb{E}[(V^{{\bm{\theta}}_L}({\bm{s}}_L^{n}) - \sum\limits_{l = 0}^{T-n}{\gamma_L}^{l}R_L^{(n+l)})^2].
\end{equation}
After updating the policy and value function parameters, the lower-tier agent generates a reference decision trajectory in the future finite steps based on the updated policy $\bar{\pi}_L$, and sends the reference decision trajectory to the higher-tier agent as a reference to update its policy.

In the second phase, after receiving the reference decision trajectory $\tau_L(\bar{\pi}_L)$, the higher-tier agent can improve its policy through performing the one-step rollout policy based on the given lower-tier policy $\bar{\pi}_L$, i.e,
{\begin{align}\nonumber
\bar{\pi}_H &= \mathop{\arg\max}\limits_{\pi_H} \mathbb{E}[R_H({\bm{s}}_H^k, \pi_H({\bm{s}}_H^k); \tau_L(\bar{\pi}_L)) \\
&\quad + \gamma_H V^{\pi_H}(h({\bm{s}}_H^k, \pi_H, \bar{\pi}_L); \bar{\pi}_L)],
\end{align}}
where $h(\cdot)$ is the state transition function at the higher-tier, i.e., ${\bm{s}}_H^{k+1} = h({\bm{s}}_H^k, \pi_H, \pi_L)$. Since ${\gamma_H}^{\bar{n}}V^{\pi_H}({\bm{s}}_H^{k + \bar{n}}; \bar{\pi}_L) \rightarrow 0$ with a large $\bar{n}$, the $\bar{n}$-rollout policy according to $\pi_H$ can be adopted to approximated $V^{\pi_H}({\bm{s}}_H^k; \bar{\pi}_L)$, i.e., $V^{\pi_H}({\bm{s}}_H^k; \bar{\pi}_L) \approx \sum_{p= 0}^{\bar{n} - 1}{\gamma_H}^{p}R_H({\bm{s}}_H^{k+p}, \pi_H({\bm{s}}_H^{k+p}); \tau_L(\bar{\pi}_L))$. However, a small estimation-step $\bar{n}$ may lead to an inaccurate estimate of the value function. To enhance the approximation performance, in this paper, one DNN with parameters ${\bm{\theta}}_H$ is adopted to estimate the tail value $V^{{\bm{\theta}}_H}({\bm{s}}_H^{k+\bar{n}})$ of the higher-tier agent. Similarly, the value function parameters are also updated through minimizing the MSE loss defined in \eqref{eq:critic_loss}. In this case, the approximated value at state ${\bm{s}}_H^k$ under the updated lower-tier policy $\bar{\pi}_L$ can be rewritten as
{\begin{align}\label{eq:rollout}\nonumber
V^{\pi_H}({\bm{s}}_H^k, \bar{\pi}_L)& = \sum\limits_{p = 0}^{\bar{n} - 1}{\gamma_H}^{p}R_H({\bm{s}}_H^{k+p}, \pi_H({\bm{s}}_H^{k+p}); \tau_L(\bar{\pi}_L))\\
&\quad + {\gamma_H}^{\bar{n}}V^{{\bm{\theta}}_H}({\bm{s}}_H^{k+\bar{n}}).
\end{align}}
Through updating policies at different agents iteratively, both the higher-tier and lower-tier value functions can be improved over original ones.
\begin{myPro}
The monotonic policy improvement can be achieved for either higher-tier agent or lower-tier agent after each policy updating stage in the proposed collaborative DRL scheme.
\end{myPro}

\begin{proof}
Please refer to Appendix B.
\end{proof}

Since the proposed scheme is an online learning scheme, each agent needs to updates its value function parameters based on the decision experiences (i.e., a tuple of state, action, reward) stored in its replay memory. As shown in Fig. \ref{fig:mmdp} and Fig. \ref{fig:time_diagram}, at the beginning of the higher-tier agent's $k$-th control cycle, the higher-tier agent visits a state ${\bm{s}}_H^{k}$ from the environment, and makes decision ${\bm{a}}_H^{k}$. Then, the higher-tier agent needs to send its current state and action and last reward to the lower-tier agent. In the subsequent $T$ time slots (i.e., from $n = kT$ to $n = kT+(T-1)$), the lower-tier agent visits a state ${\bm{s}}_L^{i}$, makes a decision ${\bm{a}}_L^{i}$, and receives corresponding reward $R_L^{i}$ at each time slot. At the end of each time slot, the lower-tier agent updates its value and policy function parameters based on its local sampled decision experiences and received decision experiences of the higher-tier agent. After the lower-tier agent updates its policy at time slot $kT + (T-1)$, it generates a reference decision trajectory of future finite steps with its updated policy parameters ${\bm{\varphi}}_L$, which is sent to the higher-tier agent. Next, the higher-tier agent obtains corresponding reward $R_H^k$ for $k$-th control cycle, and a new state ${\bm{s}}_H^{k+1}$ for making decision in $(k+1)$-th control cycle. Then, the higher-tier agent stores its decision experience $\{ {\bm{s}}_H^{k},{\bm{a}}_H^{k}, R_H^{k}\}$ to update its value function parameters. Based on the new state ${\bm{s}}_H^{k+1}$ and received reference decision trajectory, the higher-tier agent selects action ${\bm{a}}_H^{k+1}$ with respect to the maximum approximated value obtained through the $\bar{n}$-step rollout algorithm, and then the lower-tier agent visits states and makes decisions in the subsequent $T$ time slots.

\begin{figure}[t]
\centering
\includegraphics[scale = 0.36]{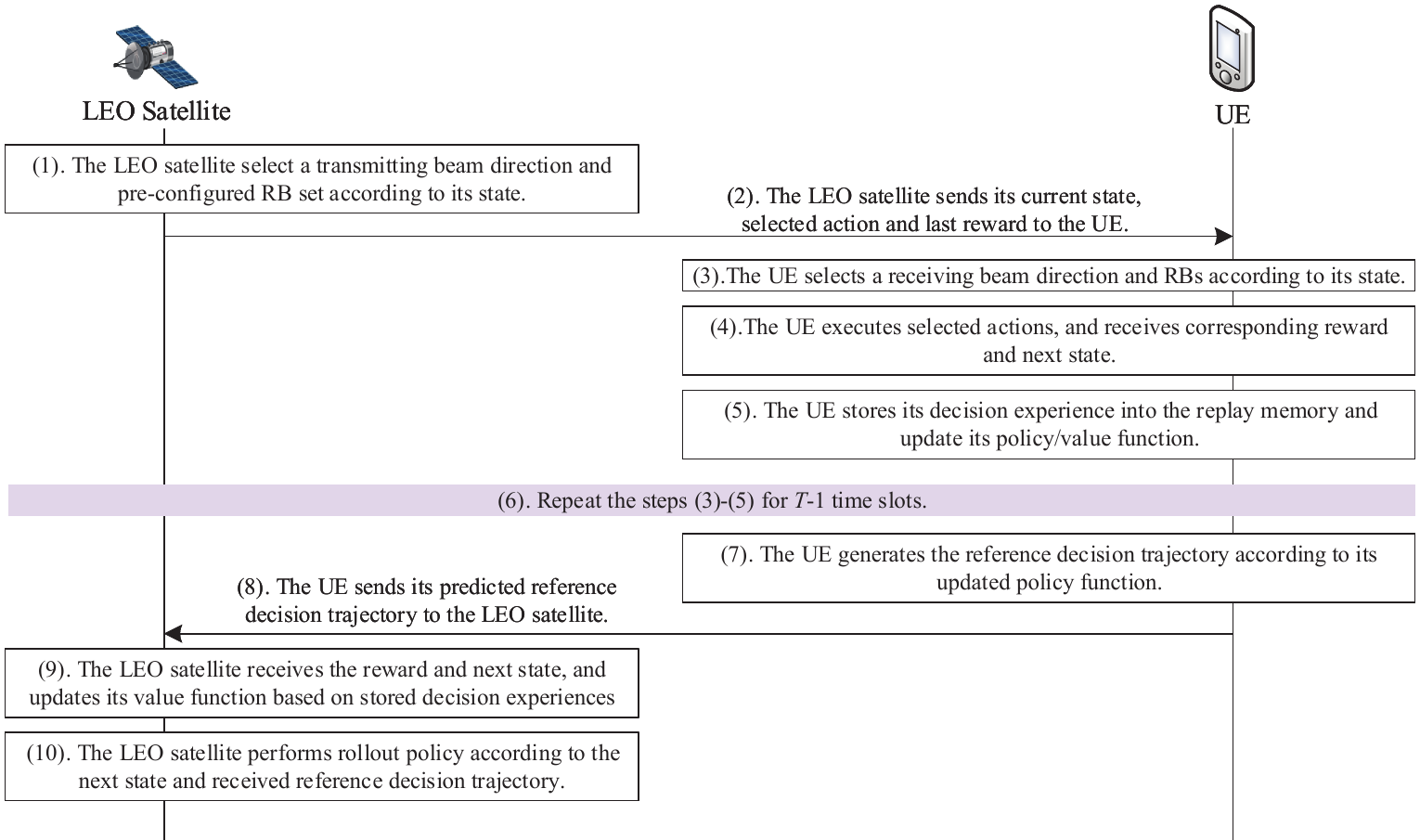}
\caption{The time diagram of the proposed two-time-scale collaborative DRL scheme in this paper.}\label{fig:time_diagram}
\end{figure}

\begin{algorithm}
{\normalsize
\begin{algorithmic}[1]
\STATE \textbf{Initialize Stage}:
\STATE Lower-tier agent (i.e., the ground UE) constructs two DNNs with randomly initialized parameters, and a replay memory ${\cal{D}}_L$.
\STATE Higher-tier agent (i.e., the LEO satellite) constructs one DNN with randomly initialized parameters, and a replay memory ${\cal{D}}_H$.
\STATE Higher-tier agent initializes its action randomly.
\STATE \textbf{Training Stage}:
\REPEAT
\STATE Higher-tier agent sends its current state-action pair $({\bm{s}}_H^{k}, {\bm{a}}_H^{k})$ and last reward ${R}_H^{k-1}$ to the lower-tier agent.
\FOR {$i = 0$ to $T-1$}
\STATE Lower-tier agent obtains action ${\bm{a}}_L^i$ based on the lower-tier agent's current state and constructed DNN-based policy function.
\STATE Lower-tier agent executes the action and receives an instantaneous reward from the environment.
\STATE Lower-tier stores the newly obtained decision experience $\{ {\bm{s}}_L^i,{\bm{a}}_L^i, R_L^i\}$ into its replay memory ${\cal{D}}_L$.
\STATE Lower-tier agent samples decision experiences from ${\cal{D}}_L$ to calculate the losses based on \eqref{eq:policy_loss} and \eqref{eq:critic_loss}.
\STATE Lower-tier agent updates policy and value function parameters ${\bm{\varphi}}_L$ and ${\bm{\theta}}_L$ through minimizing the losses.
\ENDFOR
\STATE Lower-tier agent generates the reference decision trajectory $\tau_L(\bar{\pi}_L)$ according to the updated policy.
\STATE Higher-tier agent receives the reward and reference decision trajectory from the lower-tier agent and visits a new state ${\bm{s}}_H^{k+1}$.
\STATE Higher-tier agent stores the decision experience $\{ {\bm{s}}_H^{k},{\bm{a}}_H^{k}, R_H^{k}\}$ into its replay memory ${\cal{D}}_H$.
\STATE Higher-tier agent updates ${\bm{\theta}}_H$ through minimizing the MSE over sampled decision experiences from ${\cal{D}}_H$.
\STATE Higher-tier agent obtains decision ${\bm{a}}_H^{k+1}$ through performing rollout policy in \eqref{eq:rollout} based on $\tau_L(\bar{\pi}_L)$.
\UNTIL{The update of the DNN parameters at each tier is less than a given error threshold $\epsilon$.}
\end{algorithmic}}
\caption{Proposed Two-time-scale Collaborative DRL Scheme}
\end{algorithm}

Particularly, in the updating stage of the proposed scheme, the communication overheads are incurred by two types of signaling, i.e., the decision experience of the higher-tier agent (namely, LEO satellite) and the reference decision trajectory of the lower-tier agent (namely, UE). Since the number of elements of the LEO satellite's state, action and reward are $3+T$ and $2+M$ and $1$, respectively, the data size of one state-action pair of the LEO satellite can be denoted as $(6+T+M)D_e$, where $D_e$ is the data size of one element. Additionally, since the reference decision trajectory is composed of $\bar{n}T$ actions of the lower-tier agent, the data size of the reference decision trajectory can be denoted as $(2+M)\bar{n}TD_e$. Therefore, the overall communication overhead can be denoted as $(6+T+M)D_e + (2+M)\bar{n}TD_e  = (6+T+M + (2+M)\bar{n}T)D_e$. This indicates that the overall communication overhead is dominated by the higher-tier agent's control cycle $T$, the number of estimation steps $\bar{n}$ and the number of RBs $M$. In \textbf{Algorithm 1}, the concrete procedure of the proposed scheme is summarized.

\subsection{Convergence Analysis}

In this section, we present the convergence analysis and derive the bounds of the convergence error and convergence time of our proposed two-time-scale collaborative DRL scheme. Since the optimal policies $\pi^*_H$ and $\pi^*_L$ correspond to the optimal value functions $V^{\pi_H^{*}}(s_H; \pi^{*}_L)$ and $V^{\pi_L^{*}}(s_L; \pi^{*}_H)$, respectively, the convergence of the value function is equivalent to the convergence of the policy. Consequently, to analyze the convergence of the policy, we concentrate on analyzing the convergence property of the value function instead. With the spirit of the two-time-scale model in \cite{Borkar1997}, value functions at different tiers can be regarded as being updated at the same time scale through introducing different updating rates (i.e., $a(n)$ in \eqref{eq:h_tier} and $b(n)$ in \eqref{eq:l_tier} in the following \textbf{Proposition 2}) to indicate different updating speeds incurred by different control cycles, and thus the general updating rules of agents at different tiers can be rewritten below.

\begin{myPro}
With uniform state distributions for agents at different tiers, the updating rules of different agents can be rewritten as random processes of value functions, i.e,
{\begin{align}\nonumber
 x(n+1)
& = V^{\pi_H^{n+1}}({\bm{s}}_H; \pi^{n+1}_L) \\ \nonumber
& = V^{\pi^n_H}({\bm{s}}_H; \pi^n_L) + a(n)\{\max\limits_{\pi_H}(R_H({\bm{s}}_H, \pi_H({\bm{s}}_H);\\ \nonumber
&\quad \tau_L(\pi_L^{n+1})) + \gamma_HV^{\pi_H^n}(h({\bm{s}}_H, \pi_H, \pi^{n+1}_L), \pi^{n}_L))\\ \label{eq:h_tier}
& \quad - V^{\pi^n_H}({\bm{s}}_H; \pi^n_L) + \beta_H^{n+1}\},\quad \forall {\bm{s}}_H \in {\cal{S}}_H, \\ \nonumber
 y(n+1)
& = V^{\pi^{n+1}_L}({\bm{s}}_L; \pi^{n+1}_H)\\ \nonumber
& = V^{\pi^n_L}({\bm{s}}_L; \pi^n_H) + b(n)\{R_L({\bm{s}}_L, \pi_L^{n+1}({\bm{s}}_L);\\ \nonumber
&\quad \tau_H(\pi_H^{n+1})) + \gamma_LV^{\pi^{n}_L}(g({\bm{s}}_L, \pi^{n+1}_H, \pi^{n+1}_L); \pi^{n}_H) \\ \label{eq:l_tier}
& \quad - V^{\pi^n_L}({\bm{s}}_L; \pi^n_H) + \beta_L^{n+1}\}, \quad \forall {\bm{s}}_L \in {\cal{S}}_L,
\end{align}}
where $\beta_H^{n+1} =  \max\limits_{\pi_H}\sum_{k = 0}^{\bar{n} - 1}{\gamma_H}^{k}R_H({\bm{s}}_H^{k}, \pi_H({\bm{s}}_H^{k})); \tau_L(\pi^{n+1}_L))\\ - \max\limits_{\pi_H}V^{\pi_H}({\bm{s}}_H; \pi^{n+1}_L)$ with ${\bm{s}}_H^{0} = {\bm{s}}_H$,  and $\beta_L^{n+1} = V^{{\pi}_L^{n+1}}({\bm{s}}_L; \hat{\pi}_H^{n+1}) - V^{\pi_L^{n+1}}({\bm{s}}_L; \pi^{n+1}_H)$ with $\hat{\pi}_H^{n+1} = f(\pi_L^{n+1})$,
, $g(\cdot)$ is the state transition function at the lower-tier and $a(n) > 0$ and $b(n) > 0$ are step sizes and $\frac{a(n)}{b(n)} \to \infty$, i.e., $a(n) = \frac{1 - (1 - \alpha_L)^{n}(1- (1 -\alpha_H))^{\lfloor\frac{n}{T}\rfloor}}{\lfloor\frac{n}{T}\rfloor} \approx \frac{1}{1 + n\log{n}}$ and $b(n) = 1 - (1 - \alpha_L)^n(1 - \alpha_H)^{\lfloor\frac{n}{T}\rfloor} \approx \frac{1}{n}$, which satisfy the Robbins–Monro conditions \cite{Borkar1997}, i.e., $\sum_{n = 0}^{\infty}a(n) = \sum_{n = 0}^{\infty}b(n) = \infty$ and $\sum_{n = 0}^{\infty}a(n)^2 = \sum_{n = 0}^{\infty}b(n)^2 < \infty$.
\end{myPro}

From the above updating rules, we can observe that $\pi_H$ and $\pi_L$ can be regarded as converged policies for value functions when the second terms in \eqref{eq:h_tier} and \eqref{eq:l_tier} approach zero. To illustrate the variations of these terms, we first adopt $G_H(x(n), y(n))$ and $G_L(x(n), y(n))$ to indicate the target values over $ V^{\pi^n_H}({\bm{s}}_H; \pi^n_L)$ and $ V^{\pi^n_L}({\bm{s}}_L; \pi^n_H)$, respectively, i.e.,
{\begin{align}\nonumber
G_H(x(n), y(n)) & = \max\limits_{\pi_H}(R_H({\bm{s}}_H, \pi_H({\bm{s}}_H); \tau_L(\pi_L^{n+1})) \\
&\quad + \gamma_HV^{\pi_H}(h({\bm{s}}_H, \pi_H,
\pi^{n+1}_L); \pi^{n}_L),\\ \nonumber
G_L(x(n), y(n)) & = R_L({\bm{s}}_L, \pi_L^{n+1}({\bm{s}}_L); \tau_H(\pi_H^{n+1})) \\
&\quad + \gamma_LV^{\pi^{n}_L}(g({\bm{s}}_L, \pi^{n+1}_H, \pi^{n+1}_L); \pi^{n}_H).
\end{align}}
Then, we can provide some valid properties of these stochastic terms (namely, $G_H(x(n), y(n))$, $G_L(x(n), y(n))$, $\beta_H^n$ and $\beta_L^n$) as the basis for the subsequent analysis.

\begin{myPro}
$G_H(x(n), y(n))$  and  $G_L(x(n), y(n))$ are Lipschitz with regard to the maximum norm, and $\{\beta_H^{n}\}$ and $\{\beta_L^{n}\}$ are martingale difference sequences with regard to increasing $\delta$-fields ${\cal{F}}_n \triangleq \{V^{\pi^m_H}({\bm{s}}_H; \pi^m_L), V^{\pi^m_L}({\bm{s}}_L; \pi^m_H), \beta_H^m, \beta_L^m, m \le n\}$, $n \ge 0$, satisfying $\sum_{n=0}^{\infty}a(n)\beta_H^n  < \infty$ and $\sum_{n=0}^{\infty}b(n)\beta_L^n < \infty$.
\end{myPro}

\begin{proof}
Please refer to Appendix C.
\end{proof}

After proving the martingale and Lipschitz property of the stochastic terms in \eqref{eq:h_tier} and \eqref{eq:l_tier}, the converged two-time-scale sequences should satisfy the assumptions in \textbf{Lemma 2}, which is provided in the following.

\begin{myLem}
The two-time-scale iterates $\{x_n, y_n\}$ following \cite{Borkar1997} converge to $\{\lambda(y^{*}), y^{*}\}$, almost surely, under the following assumptions: 1) \textbf{Assumption 1}: $\dot{x}(t) = h(x(t), y)$ has a unique global asymptotically stable equilibrium $\lambda(y)$, where $\lambda(\cdot)$ is a Lipschitz map. 2) \textbf{Assumption 2}: $\dot{y}(t) = g(\lambda(y(t)), y(t))$ has a unique global asymptotically stable equilibrium $y^*$. 3) \textbf{Assumption 3}: $\sup_n[\lVert x_n \rVert_{\infty} + \lVert y_n \rVert_{\infty}] < \infty$, almost surly.
\end{myLem}

Based on the assumptions in \textbf{Lemma 2}, policies at different tiers obtained based on updating rules in \eqref{eq:h_tier} and \eqref{eq:l_tier} can achieve convergence.
\begin{myTheo}
$\{\pi_H^n, \pi_L^n\}$ converge to $\{\pi_H^*, \pi_L^*\}$ almost surly.
\end{myTheo}

\begin{proof}
Please refer to Appendix D.
\end{proof}

In addition to providing the existence theorem of converged sequences, we also give the convergence error bounds under the $\infty$-norm in the following theorem.
\begin{myTheo}
\vspace{-0.5em}
After $n$ training time slots, with probability at least $1 - \delta$, the error bounds of updated value functions with respect to the optimal value functions (i.e., $x^{*}$ and $y^{*}$) are characterized as,
{\begin{align}
\lVert x^{*} - x(n) \rVert_{\infty} & = \frac{4R_H^{max}}{1 - \gamma_H}(\frac{1}{n(1 - \gamma_H)} +2\sqrt{\frac{2}{n}\ln\frac{2\lvert{\cal{S}}_H\rvert\lvert{\cal{A}}_H\rvert}{\delta}}), \label{eq:x_error} \\
\lVert y^{*} - y(n) \rVert_{\infty} & = \frac{4R_L^{max}}{1 - \gamma_L}(\frac{1}{n(1 - \gamma_L)} +2\sqrt{\frac{2}{n}\ln\frac{2\lvert{\cal{S}}_L\rvert\lvert{\cal{A}}_L\rvert}{\delta}}). \label{eq:y_error}
\end{align}}
\end{myTheo}

\begin{proof}
Please refer to Appendix E.
\end{proof}

The above theorem shows that the error bound at each time scale decreases with the increase of the number of training time slots, and increases with the dimension of the decision space. Therefore, the bound of the convergence time of our proposed scheme can be derived.

\begin{myCor}
Toward implementing $\epsilon$ error for each sequence with probability greater than $1 - \delta$, the bound of the convergence time is determined by the maximum order of the number of training slots of different sequences, i.e., $\max\{{\cal{O}}(\frac{{R_H^{max}}^2}{{\epsilon^2}(1 - \gamma_H)^2}\ln\frac{\lvert{\cal{S}}_H\rvert\lvert{\cal{A}}_H\rvert}{\delta}) ,  {\cal{O}}(\frac{{R_L^{max}}^2}{{\epsilon^2}(1 - \gamma_L)^2}\ln\frac{\lvert{\cal{S}}_L\rvert\lvert{\cal{A}}_L\rvert}{\delta})\}$.
\end{myCor}

\begin{proof}
Let the right-hand sides of \eqref{eq:x_error} and \eqref{eq:y_error} equal to $\epsilon$, and the bound of the convergence time in \textbf{Corollary 1} can be derived.
\end{proof}

From \textbf{Corollary 1}, the time complexity of our proposed scheme is dominated by the dimension of the decision space. Particularly, after the convergence is achieved, the resource optimization decisions for the LEO satellite and UE can be directly obtained with the converged DRL models. Therefore, the proposed scheme can be trained off-line and then deployed in the practical NTNs.

\section{Performance and Evaluation}

\subsection{Simulation Settings}

To evaluate the performance of our proposed scheme, consider the LEO satellite with $K = 100$ RBs, which are divided into $5$ RB groups, to service a ground UE with dynamic receiving rate demands, and RBs in \textbf{Optimization 1} are allocated in the unit of RB group. The pathloss and channel gain are generated according to 3GPP TR 38.901 \cite{3GPP2019} and TR 38.821 \cite{3GPP2021}. The ephemeris of the LEO satellite is generated by the Satellite Communication Box of MATLAB, and the minimum elevation angle between the LEO satellite and UE is $\pi/6$. Additionally, the receiving rate demand of the UE follows a Poisson distribution with a mean of $2$, and the basic unit of the rate demand is $10$ megabits per slot. For DRL models, one fully-connected DNN consisting of two layers with 300 and 200 neurons in each layer, respectively, is deployed at the LEO satellite and UE to approximate the value function. Additionally, one fully-connected DNN consisting of three layers with 400, 300 and 200 neurons in each layer, respectively, is adopted as the actor network at the UE side. Particularly, two separate layers (both with 200 neurons) are utilized to output the beam direction and RB allocation action distributions at the UE side, respectively. Besides, when calculating the beam direction of the LEO satellite/UE, the unit adjustment angle $\Delta$ is set as 5 degree. The remaining main simulation parameters are summarized in Table I.

%\addtolength{\topmargin}{0.04in}
\begin{table}\label{tab:2}
\centering
\begin{threeparttable}
\caption{Parameters for simulation.}
\vspace{-1em}
\centering
\begin{tabular}{ll}
\hline
  Parameters & Value \\
  \hline
  LEO transmission power & 30 dBW \\
  LEO/UE antenna gain  &  30 dBi \\
  Carrier frequency & 4 GHz (S band)\tnote3 \\
  Bandwidth of each RB & 180 kHz \\
  Boltzmann constant & $1.380649 \times 10^{-23} J/K $\\
  Noise temperature & 290 K \\
  Parameter optimizer &  Adam \\
  Learning rate &  0.001 \\
  Discounting factor of LEO/UE & 0.99 \\
  Replay memory size of LEO/UE & 1200/9600\\
  \hline
 \end{tabular}
 \begin{tablenotes}
     \item[3]  This paper does not consider the interferences from terrestrial BSs, which will be investigated in the future.
   \end{tablenotes}
 \end{threeparttable}
\end{table}

To evaluate the performance of the proposed DRL-based algorithm, three benchmarking schemes are considered: single-estimation, independent scheme and traditional separated optimization architecture. In the single-estimation scheme, the UE only aims to improve its advantage instead of global advantage. Besides, in the independent scheme, the LEO satellite and UE make decisions independently. The traditional separated optimization architecture includes a family of traditional optimization schemes, in which beam and RB allocation optimizations are performed sequentially. Specifically, in the beam optimization stage, the brute-force searching (BFS) scheme (in which beam-sweeping is executed with $10$ degree in each angular space) and periodic beam update (PBU) schemes in \cite{Zhao2021} (which is designed based on the geometric relationship of the LEO-UE pair) are considered. Furthermore, we consider the greedy scheme, fixed allocation scheme and multi-armed bandit (MAB) scheme for RB allocation optimization. In the greedy and MAB schemes, the LEO satellite uses all the RB groups in the RB pool to service the UE, which accesses with the greedy or MAB manner to satisfy its receiving rate demand. While in the fixed allocation scheme, the LEO satellite always allocates a fixed number of RB groups to UE randomly. Particularly, the BFS with the greedy allocation scheme can achieve the upper bound performance (i.e., the ideal performance).

\subsection{Results and Analysis}

\subsubsection{Performance with Different Numbers of Estimated Steps $\bar{n}$}
Since the LEO satellite's policy is mainly determined by the number of estimation steps, we first evaluate the moving average reward and throughput performances of the proposed scheme under different numbers of estimation steps $\bar{n}$ in Fig. \ref{fig:estimation}. Specifically, the moving average result at slot $n$ is calculated through the average operation over the original result from slot $n - N_m + 1$ to slot $n$, and the average window lengths (namely, $N_m$) for the reward and throughput performances are $10,000$ and $10$, respectively. Particularly, in simulations, one episode is defined as the duration in which the elevation angle between the UE and the LEO satellite is larger than the minimum setting value (i.e., $\pi/6$). From Fig. \ref{fig:reward_estimation}, we can observe that there is a poor reward performance for the curves with a small estimation value (i.e., $\bar{n} \le 4$). Instead, the proposed scheme achieves a higher reward performance when the estimation value is large (i.e., $\bar{n} \ge 6$). According to the illustration of $\bar{n}$-step rollout policy in Section IV-B, a large estimation $\bar{n}$ can lead to an accurate estimate of the higher-tier agent's value function, and thus a more stable policy improvement can be obtained with the increase of the estimation depth. Hence, the reward performance curve increases gradually when the number of estimation step $\bar{n}$ increases from $6$ to $8$. Additionally, please note that the DNN parameters in DRL schemes converge to a group of candidate parameters. Therefore, the parameters are switched between those final candidate parameters, leading to the reward performance perturbations after a period of training \cite{Bertsekas2019}. With this spirit, we adopt the standard deviation (SD) of the reward performance to justify the convergence. For reward curves with $\bar{n} = 6$ and $\bar{n} = 8$, their corresponding SDs over results of $7 \times 10^3$ slots decreases from $(215.8454, 161.7692)$ to $(208.7553, 156.9334)$ after a period of training, which indicates that the convergence of the reward performance with $\bar{n} = 8$ is superior to that with $\bar{n} = 6$. Therefore, in the following simulations, the number of estimation steps is set to $8$ to guarantee the performance. Furthermore, Fig. \ref{fig:rate_estimation} shows that the proposed scheme with a larger number of estimation steps can achieve an increasing moving average throughput performance when the LEO satellite and UE iteratively improve their policies. This also justifies that the rewards defined in \eqref{eq:h_reward} and \eqref{eq:l_reward} are effective in finding the proper beam direction and resource allocation decisions.

\begin{figure*}[htb]
    \subfigure[Reward]{	
	\begin{minipage}[b]{0.49\linewidth}
        \centering
        \includegraphics[scale=0.22]{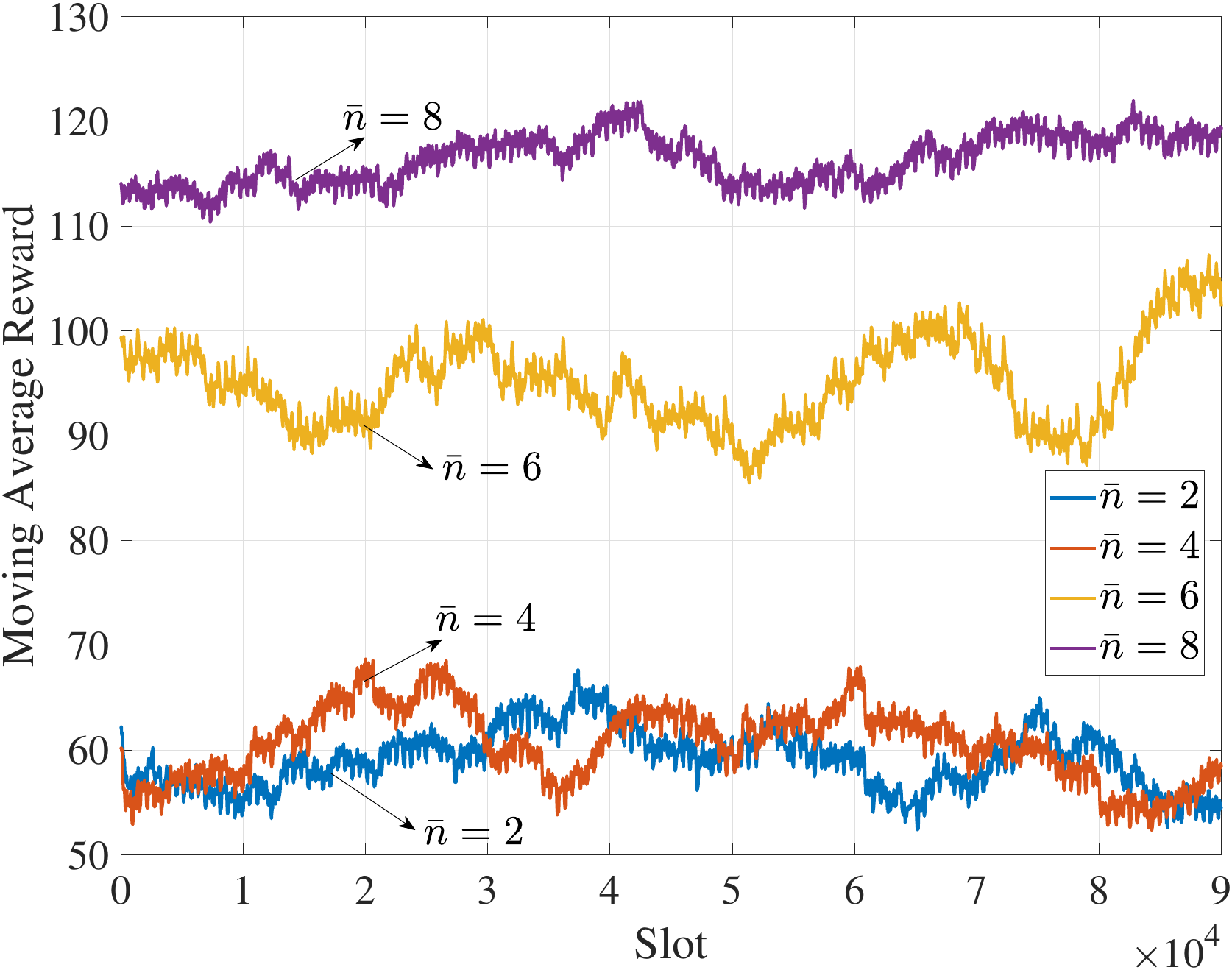}\label{fig:reward_estimation}
    \end{minipage}%
    }
    \subfigure[Throughput]{
    \begin{minipage}[b]{0.49\linewidth}
        \centering
        \includegraphics[scale=0.22]{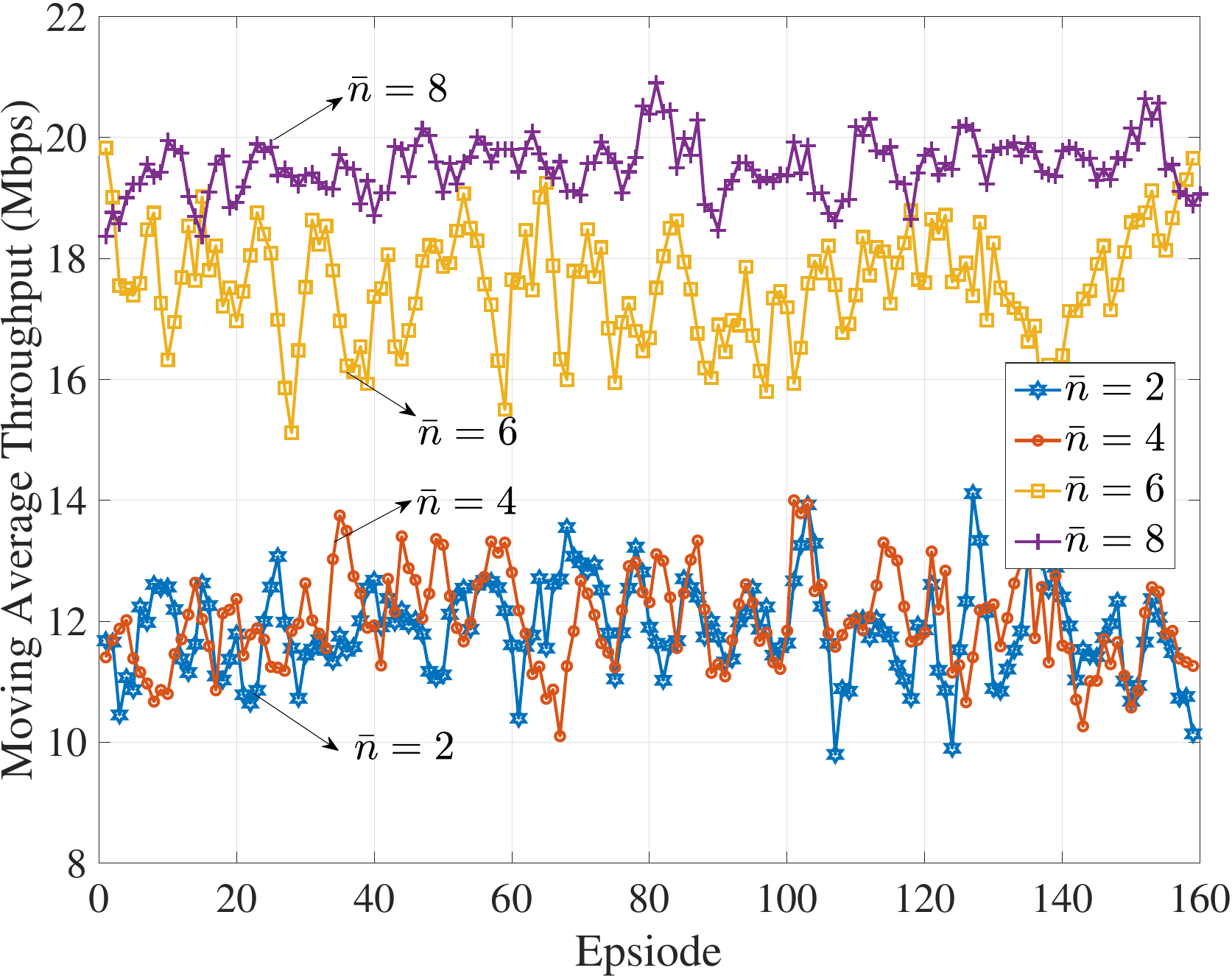}\label{fig:rate_estimation}
    \end{minipage}} \\[0.90mm]
	\caption{Moving average reward and throughput with different numbers of estimation steps $
\bar{n}$.}\label{fig:estimation}
\end{figure*}

\subsubsection{Performance with Different Numbers of Transmitting and Receiving Antennas $(N_t, N_r)$}
As illustrated in \eqref{eq:snr}, the beam direction decision is of most importance in improving the SNR of each RB. Subsequently, we evaluate the moving average reward and throughput performances of the proposed scheme with different numbers of transmitting and receiving antennas, as shown in Fig. \ref{fig:antenna}. Since a larger number of antennas leads to a higher spatial resolution, the transmitting and receiving beams will become narrower and deliver more power. In this case, the SNR at the UE will decrease significantly if the transmitting and receiving beams are not well aligned, and thus the satellite and UE normally cannot obtain a positive reward at most time according to the reward definitions in \eqref{eq:h_reward} and \eqref{eq:l_reward}. In the meantime, with the training loss in \eqref{eq:critic_loss}, a lower value function is estimated based on the obtained rewards so that the policies cannot be fully improved. Therefore, in Fig. \ref{fig:reward_antenna}, we can observe that the proposed scheme can achieve increasing reward performances with the increase of time slots when lower numbers of transmitting and receiving antennas are adopted, but the reward performance decreases with the increase of the number of antennas since the narrow beam cannot be captured easily. Additionally, for the throughput performance comparison, we also present the performance using the combination of the BFS scheme and greedy RB allocation scheme (i.e., BFS-Greedy) in Fig. \ref{fig:rate_antenna}. As aforementioned, the BFS-Greedy scheme achieves the ideal performance. With a small number of transmitting and receiving antennas (i.e., $(4^2, 8^2)$), the proposed scheme can achieve $78.5 \%$ throughput performance of the BFS-Greedy scheme. When the number of transmitting and receiving antennas increases, the throughput performance of the proposed scheme degrades to $58.9\%$ of that of the BFS-Greedy scheme. Although performance degradations exist, the searching space dimension of the BFS-Greedy scheme (i.e., $18^4$ at each time slot) is significantly larger than that of the beam action space (i.e., $7^2 \times \frac{3^2}{T}$) in our proposed scheme. Therefore, our proposed scheme can achieve a better tradeoff between throughput performance and computational complexity than that of the BFS-Greedy scheme. In the following simulations, $(N_t, N_r)$ is set as $(16^2,8^2)$.

\begin{figure*}[htb]
    \subfigure[Reward]{	
	\begin{minipage}[b]{0.49\linewidth}
        \centering
        \includegraphics[scale=0.22]{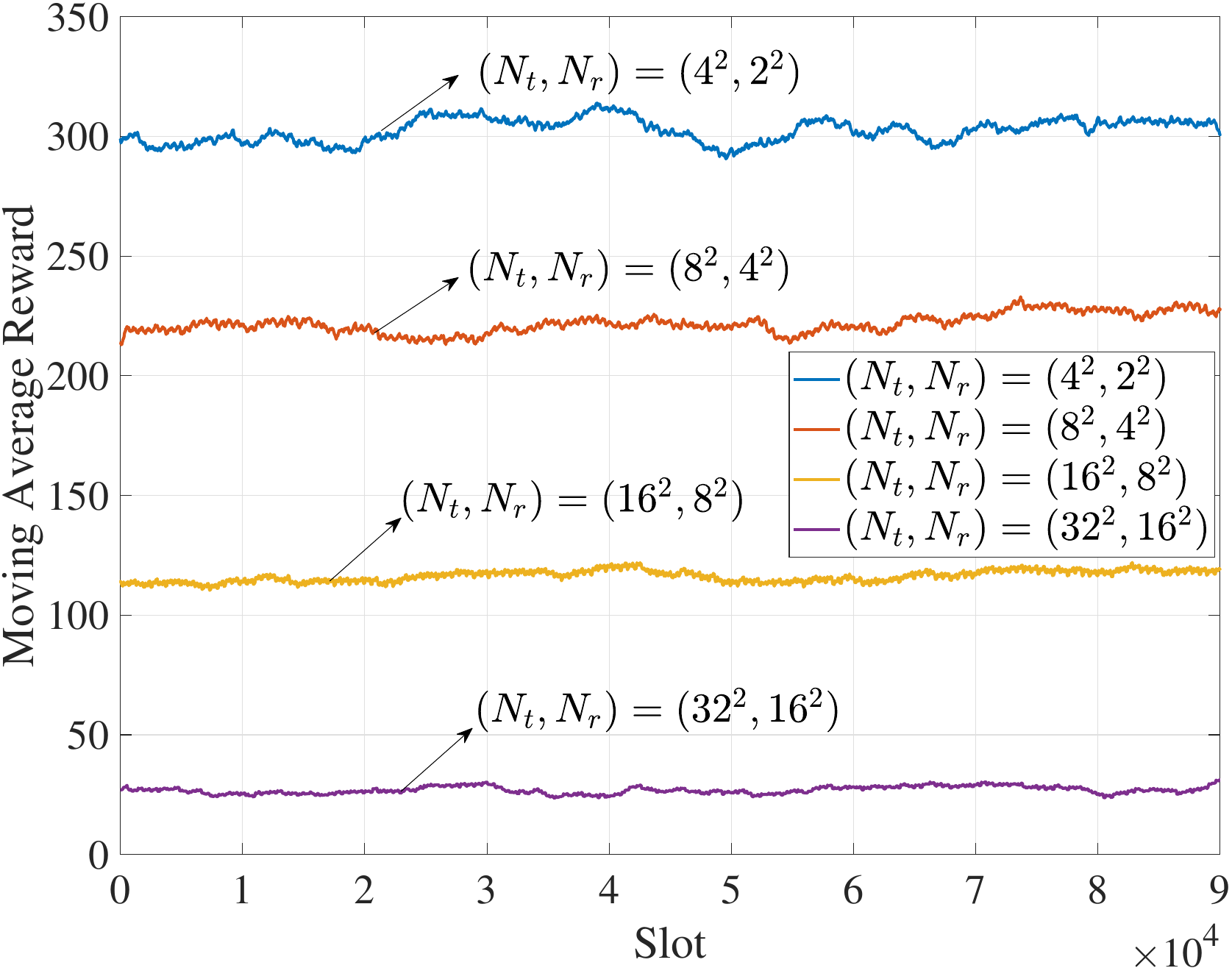}\label{fig:reward_antenna}
    \end{minipage}%
    }
    \subfigure[Throughput]{
    \begin{minipage}[b]{0.49\linewidth}
        \centering
        \includegraphics[scale=0.22]{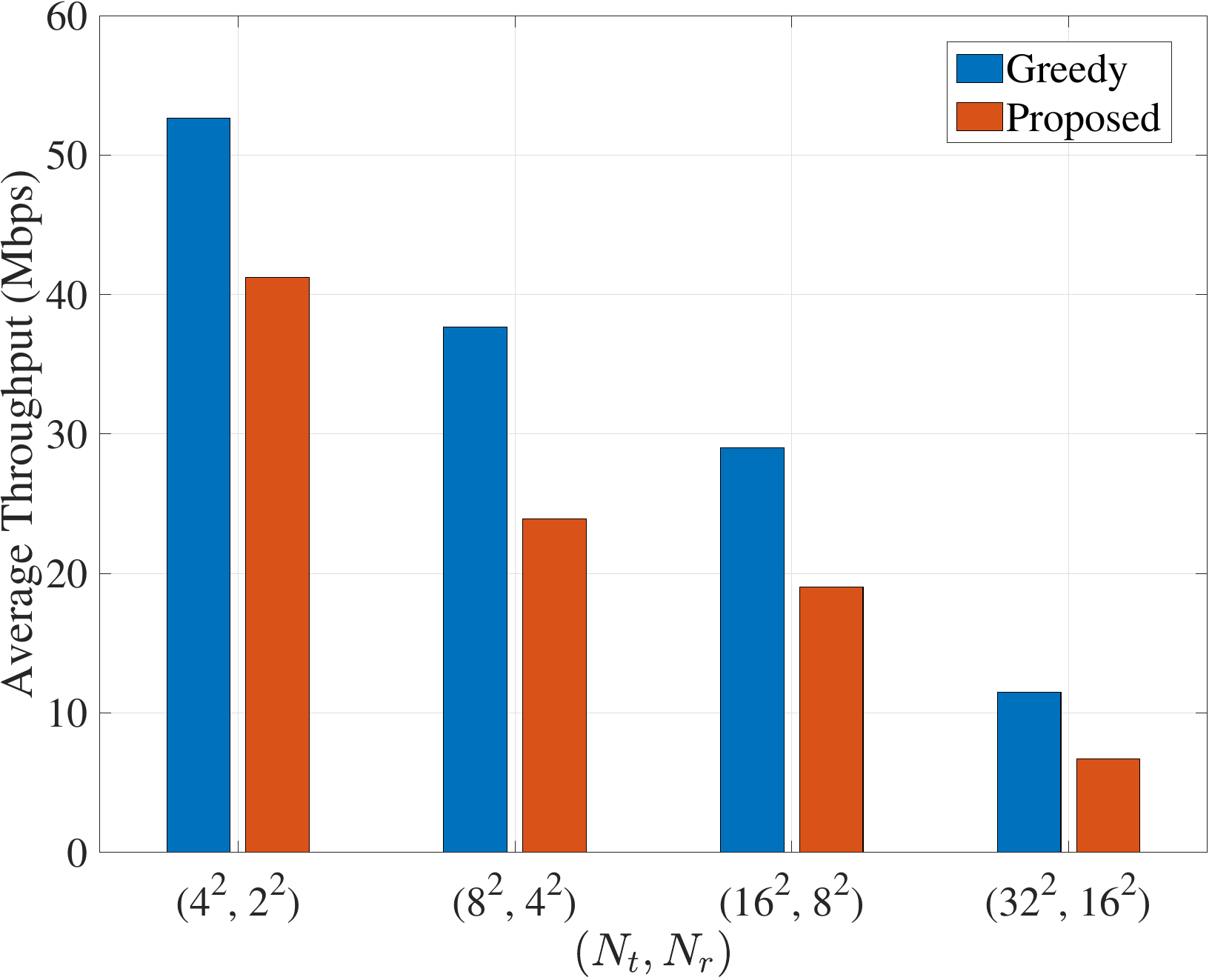}\label{fig:rate_antenna}
    \end{minipage}} \\[0.90mm]
	\caption{Moving average reward and throughput with different DRL-based algorithms.}\label{fig:antenna}
\end{figure*}

\subsubsection{Performance with Different DRL-based Schemes}
Next, to justify the effectiveness of our proposed scheme, we present the moving average reward and throughput performances of the proposed scheme and other DRL-based variants in Fig. \ref{fig:drl}. Through comparing the performance curves in Fig. \ref{fig:drl}, we can find that the proposed scheme and the single-estimation scheme can achieve a large performance enhancement over the independent scheme in terms of the reward and throughput, since each agent in the independent scheme can only exploit its own local information for decision-makings without collaborations. However, due to the differences of the control cycles, the environment of the NTN is non-stationary, leading to severe performance degradations. While in the proposed scheme and the single-estimation scheme, the sequential manner can lead to a stable policy improvement for the LEO satellite and UE since each agent fixes its policy and provides a temporarily stationary environment to the other agent. On the other hand, in Fig. \ref{fig:drl}, the performances curves of the proposed scheme are higher than those of the single-estimation scheme. This is because that, the UE estimates the policy advantages of the satellite and UE in \eqref{eq:policy_loss} so that it can not only improve the UE's value function but also enable the LEO satellite to improve the sum of value functions through performing the finite-step rollout policy in \eqref{eq:rollout}. In this case, the collaboration effects of the proposed scheme can be improved as compared with that of the single-estimation scheme, in which the UE only estimates its own policy advantage.

\begin{figure*}[htb]
    \subfigure[Reward]{	
	\begin{minipage}[b]{0.49\linewidth}
        \centering
        \includegraphics[scale=0.22]{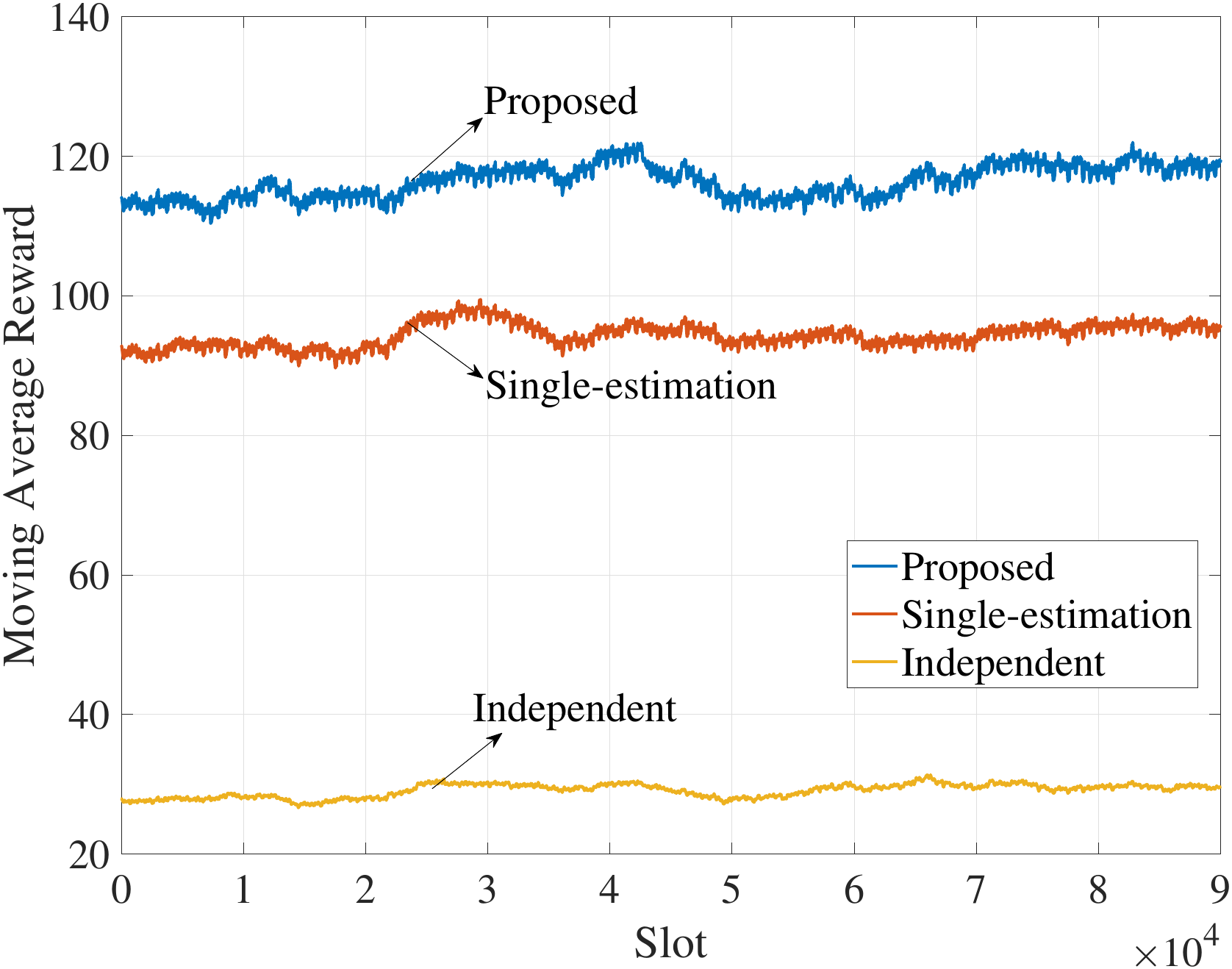}\label{fig:reward_drl}
        %\vspace{-2em}
    \end{minipage}%
    }
    \subfigure[Throughput]{
    \begin{minipage}[b]{0.49\linewidth}
        \centering
        \includegraphics[scale=0.22]{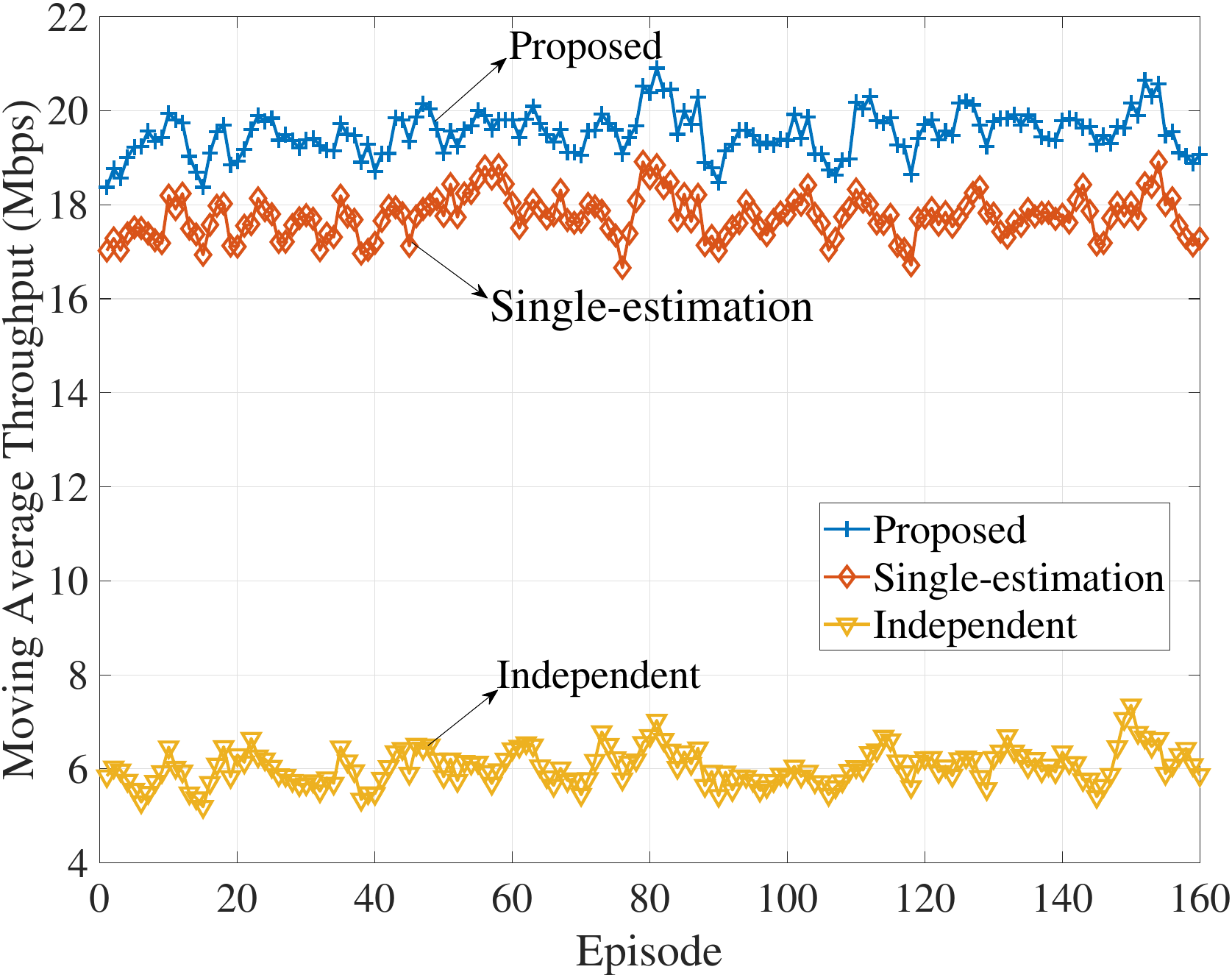}\label{fig:rate_drl}
        %\vspace{-2em}
    \end{minipage}}\\% \\[0.90mm]
	\caption{Moving average reward and throughput with different DRL-based schemes.}\label{fig:drl}
\end{figure*}

\subsubsection{Performance with Different Optimization Schemes}

\begin{figure*}[htb]
    \subfigure[Satisfactory error]{	
	\begin{minipage}[b]{0.49\linewidth}
        \centering
        \includegraphics[scale=0.22]{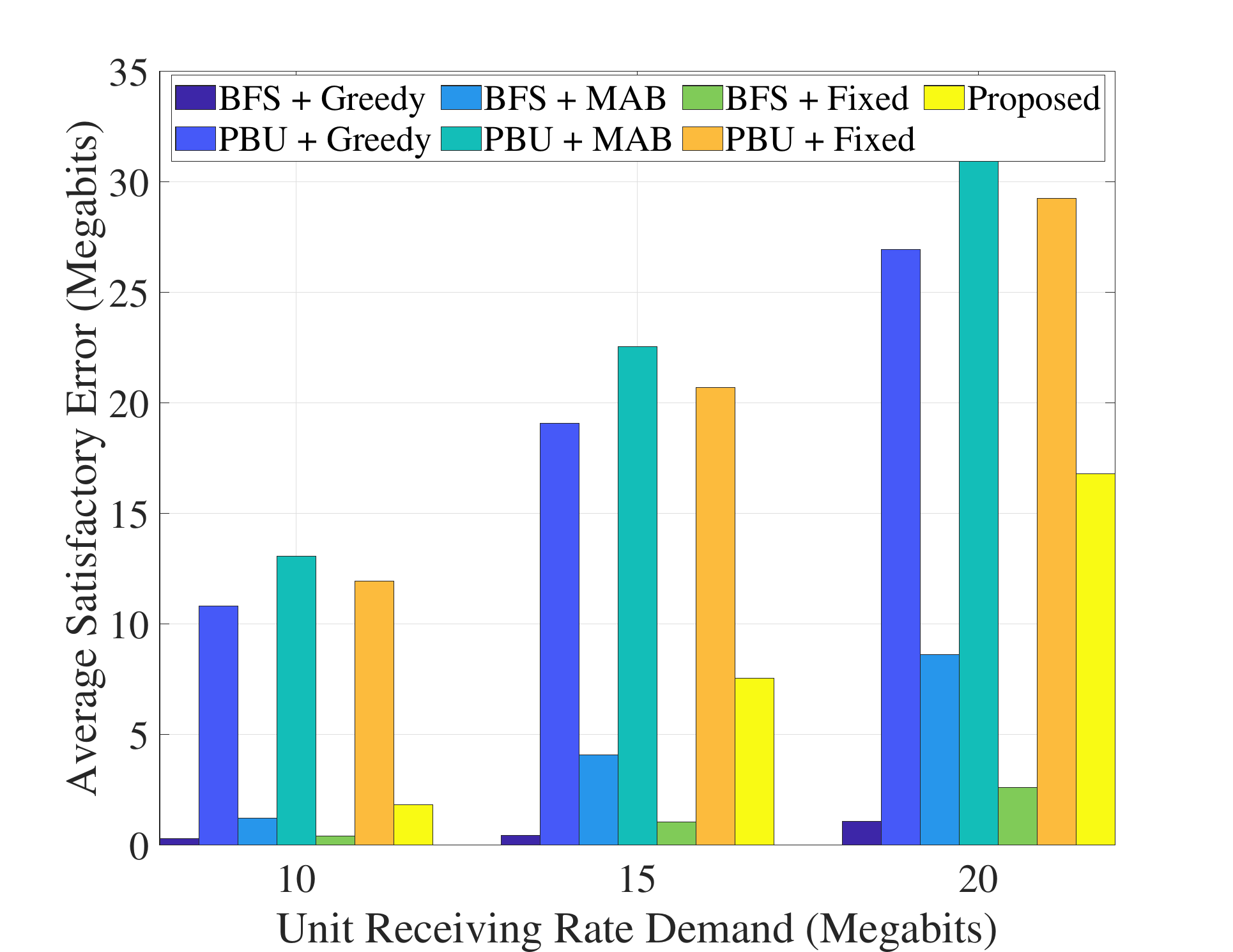}\label{fig:satisfactory}
    \end{minipage}%
    }
    \subfigure[Number of allocated RB groups]{
    \begin{minipage}[b]{0.49\linewidth}
        \centering
        \includegraphics[scale=0.22]{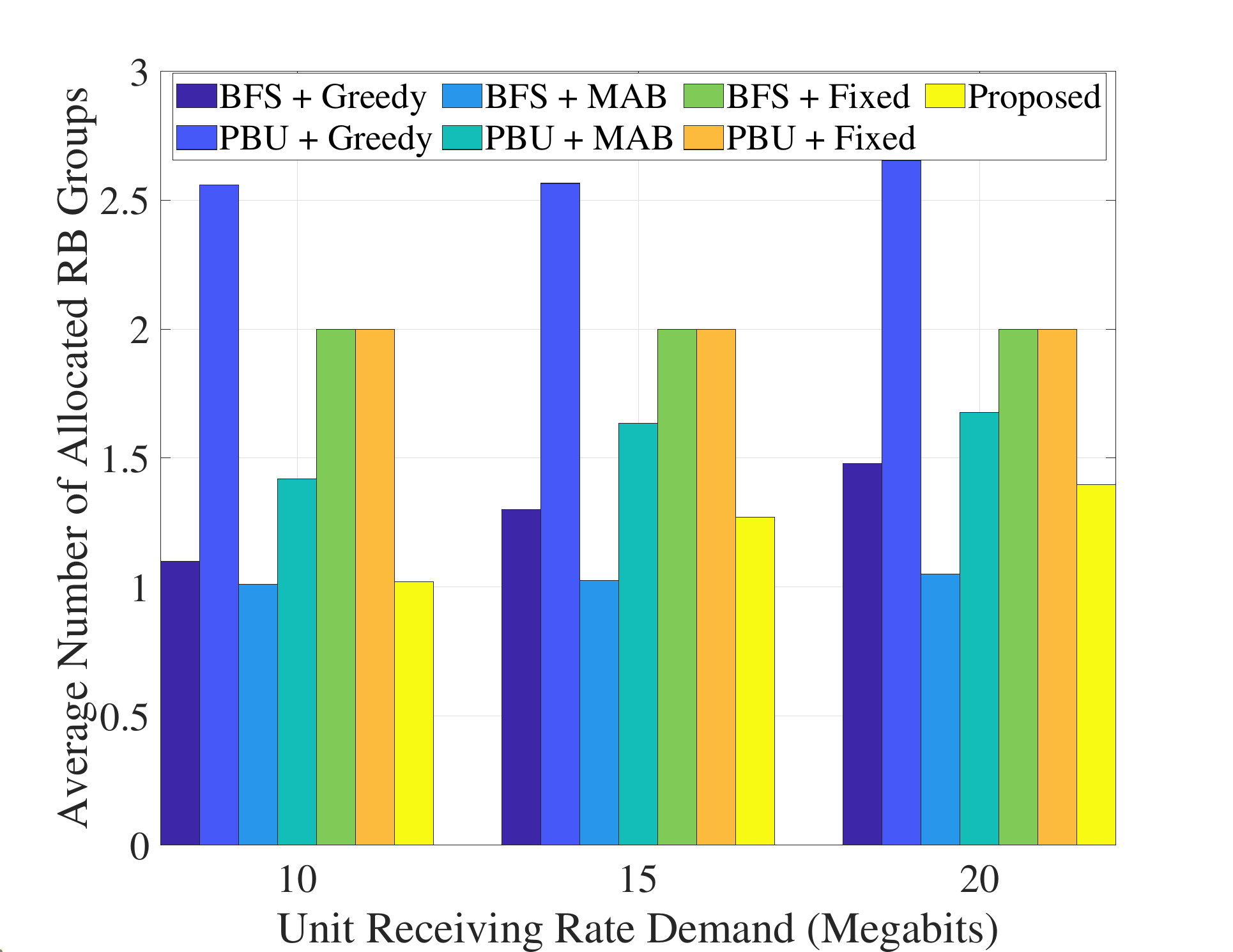}\label{fig:satisfactory_num}
    \end{minipage}} \\[0.90mm]
	\caption{Average satisfactory error and number of allocated RB groups with different schemes.}\label{fig:satis}
\end{figure*}

Finally, since the objective of the proposed scheme is to satisfy the dynamic receiving rate demand of the UE using the minimum number of RB groups, we present the average satisfactory error performances and numbers of allocated RB groups of the proposed scheme and traditional separate optimization schemes in Fig. \ref{fig:satis}. In Fig. \ref{fig:satisfactory}, the satisfactory error is defined as the absolute value of $\Omega_i$ at each time slot, and its optimal value is $0$. From Fig. \ref{fig:satisfactory_num}, we can observe that our proposed scheme achieves a significant performance enhancement over the other schemes in terms of the number of allocated RB groups since the UE reward function increases with the decrease of the number of utilized RB groups. Similarly, the MAB-based schemes also offer the performance gain in minimizing the number of allocated RB groups due to the reward function. Through comparing the performances shown in Fig. \ref{fig:satisfactory}, the satisfactory error performance achieved by the proposed scheme is close to that of the BFS-MAB scheme. Additionally, the proposed scheme outperforms the other three PBU-based schemes in minimizing the satisfactory error with a less amount of allocated RB groups. To further explore the relationship between the transmission performance and computational complexity, a weighted-sum multi-attribute utility function \cite{Senouci2016} is introduced to measure the composite score of the satisfactory error, number of allocate RB groups and computational complexity of these different schemes under different weight combinations, and the results are presented in Table \ref{tab:utility}. Four example combinations are considered, i.e., 1) the weight combination $[\frac{1}{3},  \frac{1}{3}, \frac{1}{3}]$ corresponds to the generalized scenario; 2)  the weight combination $[\frac{1}{2},  \frac{1}{4}, \frac{1}{4}]$ corresponds to the latency-sensitive scenario; 3)  the weight combination $[\frac{1}{4},  \frac{1}{2}, \frac{1}{4}]$ corresponds to the bandwidth-constrained scenario; and 4)  the weight combination $[\frac{1}{4},  \frac{1}{4}, \frac{1}{2}]$ corresponds to the energy-sensitive scenario. In PBU-based schemes, since the beam direction is calculated based on the Newton's method at each time slot, the logarithmic complexity of its maximum adjustable angular value (i.e., $\pi = 180$ degree) is adopted. In Table \ref{tab:utility}, we can observe that the proposed scheme can achieve the optimal utility (i.e., the minimum utility score) over the other schemes under different weight scenarios. Although the BPU-based schemes are with less computational complexity, there are larger satisfactory errors, and more RB groups are also required. Moreover, the BFS-Greedy scheme can achieve the minimum satisfactory error with a relatively small number of RB groups, but it is with the maximum computational complexity. Therefore, the proposed scheme can achieve a superior performance to balance the satisfactory error, utilized resources and computational complexity.

\begin{table*}
\centering
%\small
\caption{Weighted-sum Utility of Different Schemes.}
\begin{tabular}{ccccccccc}
\hline
\hline
\multirow{2}{*}{\thead{Utility \\ Coefficients}}&
\multirow{2}{*}{\thead{Unit Receiving\\ Rate Demand\\(Megabits)}}
&\multicolumn{7}{c}{Scheme}\cr\cline{3-9}
& &  \thead{BFS-\\Greedy}&\thead{PBU-\\Greedy}&\thead{BFS-\\MAB}&\thead{PBU-\\MAB}&\thead{BFS-\\Fixed}&\thead{PBU-\\Fixed}&Proposed\cr
\hline

\multirow{3}{*}{\thead{$[\frac{1}{3},  \frac{1}{3}, \frac{1}{3}]$\\(generalized scenario)}}
     & 10 & 0.2794 & 0.3642 & 0.2872 & 0.3152 & 0.3485 & 0.3405 & 0.1056\cr\cline{2-9}
     & 15 & 0.2896 & 0.3554 & 0.3027 & 0.3199 & 0.3452 & 0.3294 & 0.1590\cr\cline{2-9}
     & 20 & 0.3017 & 0.3493 & 0.3181 & 0.3110 & 0.3470 & 0.3181 & 0.2001\cr\hline
\multirow{3}{*}{\thead{$[\frac{1}{2},  \frac{1}{4}, \frac{1}{4}]$\\(latency-sensitive scenario)}}
     & 10 & 0.3539 & 0.2732 & 0.3597 & 0.2364 & 0.4057 & 0.2553 & 0.0793\cr\cline{2-9}
     & 15 & 0.3615 & 0.2666 & 0.3713 & 0.2399 & 0.4032 & 0.2471 & 0.1193\cr\cline{2-9}
     & 20 & 0.3706 & 0.2620 & 0.3829 & 0.2333 & 0.4046 & 0.2386 & 0.1501\cr\hline
\multirow{3}{*}{\thead{$[\frac{1}{4},  \frac{1}{2}, \frac{1}{4}]$\\(bandwidth-constrained scenario)}}
     & 10 & 0.2130 & 0.4026 & 0.2298 & 0.3930 & 0.2662 & 0.3984 & 0.1011\cr\cline{2-9}
     & 15 & 0.2201 & 0.3951 & 0.2545 & 0.3919 & 0.2659 & 0.3866 & 0.1701\cr\cline{2-9}
     & 20 & 0.2312 & 0.3859 & 0.2783 & 0.3793 & 0.2722 & 0.3732 & 0.2273\cr\hline
\multirow{3}{*}{\thead{$[\frac{1}{4},  \frac{1}{4}, \frac{1}{2}]$\\(energy-sensitive scenario)}}
     & 10 & 0.2713 & 0.4168 & 0.2721 & 0.3161 & 0.3736 & 0.3676 & 0.1365\cr\cline{2-9}
     & 15 & 0.2871 & 0.4045 & 0.2822 & 0.3278 & 0.3665 & 0.3547 & 0.1875\cr\cline{2-9}
     & 20 & 0.3032 & 0.4000 & 0.2932 & 0.3205 & 0.3643 & 0.3426 & 0.2227\cr
\hline
\hline
\end{tabular}\label{tab:utility}
\end{table*}

\section{Conclusion}

In this paper, we investigate the beam management and resource allocation in earth-fixed cells of NTNs, and propose a two-time-scale collaborative DRL scheme to enable LEO satellite and UE as agents with different control cycles to perform decision-making tasks collaboratively. Specifically, to obtain monotonic policy improvements, a sequential policy updating manner is adopted. Targeting at performing computations for the LEO satellite, the UE with a powerful computing capability improves its policy based on the sum advantage function of both two agents, and predicts its reference decision trajectory of future finite steps, which is sent to the LEO satellite. With the reference decision trajectory from the UE, the LEO satellite only needs to select actions through the finite-step rollout policy. Particularly, we also provide the theoretical analysis on the convergence of the proposed scheme. Simulation results show that the proposed scheme outperforms the DRL-based schemes without proper collaboration designs in terms of the throughput performance, and can effectively balance the satisfactory error, amount of allocated RBs and computational complexity over other BFS-based and PBU-based schemes. Through carefully designing the TTMDP model, the proposed two-time-scale collaborative DRL scheme can be transferred to solve multi-domain resource optimizations of non-terrestrial communications and spectrum-sharing between non-terrestrial and terrestrial communications. With the established analytical foundations of the two-time-scale two-agent DRL, this research also creates a new research roadmap and foundations toward implementing intelligent resource optimizations for multi-time-scale multi-agent scenarios of NTNs and other multi-tier heterogeneous networks such as the space-aerial-ground-integrated three-tier networks involving various satellite and UAV flying BSs.

\begin{appendices}

\section{Proof of Lemma 1}

When updating the policy of the lower-tier agent, the corresponding policy improvement of the higher-tier agent is equivalent to the difference on its advantage function \cite{Schulman2015}, i.e., $
\mathop{\mathbb{E}}_{{\bm{s}}_H^0, {\bm{a}}_H^0, \ldots}[ \sum_{k=0}^{\infty}{\gamma_H}^{k}A_{\pi_H}({\bm{s}}_H^{k}, {\bm{a}}_H^{k}; \tau_L(\pi_L))]  = -\rho_H(\pi_H, \pi'_L) + \rho_H(\pi_H, \pi_L)$,
where $s_H^0 \sim d_H^0({\bm{s}}_H^0)$, ${\bm{a}}_H^{k} \sim \pi_H(\cdot|{\bm{s}}_H^k)$ and ${\bm{s}}_H^{k+1} \sim P_H({\bm{s}}_H^{k}, \pi_H, \pi_L)$. In this proof, the $\alpha$-couple policy assumption is adopted, i.e., $P(\bar{\pi}_L \neq \pi_L|{\bm{s}}) \le \alpha_L$ and $P(\bar{\pi}_H \neq \pi_H|{\bm{s}}) \le \alpha_H$, $\forall$ ${\bm{s}} = ({\bm{s}}_L, {\bm{s}}_H)$ $\in {\cal{S}}_L \times {\cal{S}}_H$. Therefore, the probabilities of adopting $\pi'_L$ and $\pi_L$ are $1-\alpha_L$ and $\alpha_L$, respectively. The probability of selecting $\pi_L$ to improve the value function of the higher-tier agent before updating stage $k$ is $1 - (1 - \alpha_L)^{kT}$, i.e., $P(\rho_H^k > 0) = 1-(1 - \alpha_L)^{kT}$. Therefore, the expected policy improvement of higher-tier agent incurred by the policy of the lower-tier agent is bounded, i.e,
{\setlength\abovedisplayskip{3pt}
 \setlength\belowdisplayskip{3pt}
\begin{small}
\begin{align*}\nonumber
& \quad \rho_H(\pi_H, \pi_L) - \rho_H(\pi_H, \pi'_L) \\
& = \sum\limits_{{\bm{s}}_H}d_H^{\pi'_L, \pi_H}({\bm{s}}_H)\sum\limits_{{\bm{a}}_H \in {\cal{A}}_H}\pi_H({\bm{a}}_H|{\bm{s}}_H)\mathop{\mathbb{E}}\limits_{{\bm{s}}_H^{k+1}}[R_H({\bm{s}}_H, {\bm{a}}_H; \tau_L(\pi_L))\\
&\quad + {\gamma_H}V^{\pi_H}({\bm{s}}_H^{k+1};\pi'_L)- V^{\pi_H}({\bm{s}}_H^k;\pi'_L)]\\
& = L_{\pi_H}(\pi_L) - \sum\limits_{k=0}^{\infty}{\gamma_H}^k P(\rho_H^k > 0)\{\mathop{\mathbb{E}}\limits_{{\bm{s}}_H^k}[A_{\pi_H}({\bm{s}}_H^k, {\bm{a}}_H^k; \tau_L(\pi_L))|\\
&\quad \rho_H^k = 0] - \mathop{\mathbb{E}}\limits_{{\bm{s}}_H^k}[A_{\pi_H}({\bm{s}}_H^k, {\bm{a}}_H^k; \tau_L(\pi_L))|\rho_H^k > 0]\}\\
& \ge L_{\pi_H}(\pi_L) - 4(1 -(1 -\alpha_L)^T)\epsilon_H(\pi_L)\sum\limits_{k=0}^{\infty}{\gamma_H}^k(1- (1 - \alpha_L)^{kT}) \\
%& = L_{\pi_H}(\pi_L) -  4(1 -(1 -\alpha_L)^T)\epsilon_H(\pi_L)(\frac{1}{1-\gamma_H} \\
%& \quad - \frac{1}{1-(1-\alpha_L)^T\gamma_H})\\
%= L_{\pi_H}(\pi_L)- \frac{4\epsilon_H(\pi_L)(1 -(1 -\alpha_L)^T)(1-(1-\alpha_L)^T)\gamma_H}{(1-\gamma_H)(1-(1-\alpha_L)^T\gamma_H)}
 &= L_{\pi_H}(\pi_L)- \frac{4\epsilon_H(\pi_L)(1 -(1 -\alpha_L)^T)^2\gamma_H}{(1-\gamma_H)(1-(1-\alpha_L)^T\gamma_H)}.
\end{align*}
\end{small}}
Next, we concentrate on the policy improvement of the lower-tier agent. Similarly, the policy improvement part is also analyzed from the advantage function, i.e., $
\mathop{\mathbb{E}}_{{\bm{s}}_L^0, {\bm{a}}_L^0, \ldots}[\sum_{i=0}^{\infty}{\gamma_L}^{i}A_{\pi_L}({\bm{s}}_L^{i}, {\bm{a}}_L^{i}; f(\pi_H))] = \rho_L(\pi_L, f(\pi_L))-\rho_L(\pi'_L, \pi_H)$,
where ${\bm{s}}_L^0 \sim d_L^0({\bm{s}}_L^0)$, ${\bm{a}}_L^{i} \sim \pi_L(\cdot|{\bm{s}}_L^i)$ and ${\bm{s}}_L^{i+1} \sim P_L({\bm{s}}_L^{i}, \pi_L, f(\pi_L))$. Similarly, the probability of selecting $\pi_L$ and $f(\pi_L)$ is $1 - (1 - \alpha_L)^i(1 - \alpha_H)^{\lfloor\frac{i}{T}\rfloor}$, i.e.,  $P(\rho_L^i > 0) = 1 - (1 - \alpha_L)^i(1 - \alpha_H)^{\lfloor\frac{i}{T}\rfloor}$, and we have

{\setlength\abovedisplayskip{3pt}
 \setlength\belowdisplayskip{3pt}
\begin{small}
\begin{align*}
& \quad \rho_L(\pi_L, f(\pi_L)) - \rho_L(\pi'_L, \pi_H) \\
& =  L_{\pi'_L, \pi_H}(\pi_L, f(\pi_L)) - \sum\limits_{i = 0}^{\infty}\gamma_L^i P(\rho_L^i > 0)\{\mathop{\mathbb{E}} \limits_{{\bm{s}}_L^i}[A_{\pi_L}({\bm{s}}_L^i, {\bm{a}}_L^i;\\
&\quad \tau_H(f(\pi_L)))|\rho_L^i = 0] -\mathop{\mathbb{E}}\limits_{{\bm{s}}_L^i}[A_{\pi_L}({\bm{s}}_L^t, {\bm{a}}_L^i; \tau_H(f(\pi_L)))|\rho_L^i> 0]\}\\
& \ge L_{\pi'_L, \pi_H}(\pi_L, f(\pi_L)) - 2\sum\limits_{i=0}^{\infty}{\gamma_L}^i(1 - (1 - \alpha_L)^i(1 - \alpha_H)^{\lfloor\frac{i}{T}\rfloor})\\
&\quad \cdot 2(1 - (1-\alpha_L)(1 - \alpha_H)^{\frac{1}{T}})\epsilon_L(\pi_H)\\
&= L_{\pi'_L, \pi_H}(\pi_L, f(\pi_L)) - 4(1 - (1-\alpha_L)(1 - \alpha_H)^{\frac{1}{T}})\epsilon_L(\pi_H)(\\
& \quad \sum\limits_{t=0}^{\infty}{\gamma_L}^i - \sum\limits_{i=0}^{\infty}{\gamma_L}^i(1 - (1 - \alpha_L)^i(1 -  \alpha_H)^{\lfloor\frac{i}{T}\rfloor})).
\end{align*}
\end{small}}
Since
{\setlength\abovedisplayskip{3pt}
\setlength\belowdisplayskip{3pt}
\begin{small}
\begin{align*}
&\quad  \sum\limits_{i=0}^{\infty}{\gamma_L}^i(1 - \alpha_L)^i(1 - \alpha_H)^{\lfloor\frac{i}{T}\rfloor} = \sum\limits_{i=0}^{\infty}\frac{(\gamma_L(1-\alpha_L)(1- \alpha_H)^{\frac{1}{T}})^i}{(1-\alpha_H)^{(\frac{i}{T} - \lfloor\frac{i}{T}\rfloor)}} \\
& = \sum\limits_{j = 0}^{T-1}\frac{\sum\limits_{i =0}^{\infty}(\gamma_L(1-\alpha_L)(1- \alpha_H)^{\frac{1}{T}})^{iT+j}}{(1-\alpha_H)^{\frac{j}{T}}}\\ %= \sum\limits_{j = 0}^{T-1}[\frac{(\gamma_L(1-\alpha_L)(1- \alpha_H)^{\frac{1}{T}})^j}{(1-\alpha_H)^{\frac{j}{T}}}\cdot\frac{1}{1 - {\gamma_L}^T(1-\alpha_L)^T(1- \alpha_H)}]
&= \frac{1-{\gamma_L}^T(1-\alpha_L)^T}{(1 - \gamma_L(1-\alpha_L))(1 - {\gamma_L}^T(1-\alpha_L)^T(1- \alpha_H))},
\end{align*}
\end{small}}
we can derive the lower-tier agent's policy improvement bound illustrated in \eqref{eq:pi_ltier}.

\section{Proof of Proposition 1}

First, we proof that the value function of the higher-tier agent can be improved in a sequential updating manner, i.e.,
{\setlength\abovedisplayskip{3pt}
 \setlength\belowdisplayskip{3pt}
\begin{small}
\begin{align*}\nonumber
&\quad V^{\bar{\pi}_H}({\bm{s}}_H;\bar{\pi}_L) \\
& = \max\limits_{\bar{\pi}_H }\mathbb{E}[R_H({\bm{s}}_H, \bar{\pi}_H({\bm{s}}_H); \tau_L(\bar{\pi}_L)) + \gamma_HV^{\bar{\pi}_H}(h({\bm{s}}_H, \bar{\pi}_H, \bar{\pi}_L); \bar{\pi}_L) ]\\
                     & = \max\limits_{\bar{\pi}_H }\mathbb{E}[R_H({\bm{s}}_H, \bar{\pi}_H({\bm{s}}_H); \tau_L(\bar{\pi}_L)) +\gamma_HV^{\pi_H}(h({\bm{s}}_H, \bar{\pi}_H, \bar{\pi}_L); \pi_L)\\
                     &\quad -V^{\pi_H}({\bm{s}}_H; \pi_L) + V^{\pi_H}({\bm{s}}_H; \pi_L)]\\
                     & \ge \mathbb{E}[R_H({\bm{s}}_H, \pi_H({\bm{s}}_H); \tau_L(\bar{\pi}_L)) +\gamma_HV^{\pi_H}(h({\bm{s}}_H, \pi_H, \bar{\pi}_L); \pi_L)\\
                     &\quad -V^{\pi_H}({\bm{s}}_H; \pi_L) + V^{\pi_H}({\bm{s}}_H; \pi_L)]\\
                     &= \max\limits_{\bar{\pi}_L}[\rho_H(\pi_H, \bar{\pi}_L) - \rho_H(\pi_H, \pi_L)]+ V^{\pi_H}({\bm{s}}_H; \pi_L)\\
                     & \ge V^{\pi_H}({\bm{s}}_H;\pi_L), \forall {\bm{s}}_H \in {\cal{S}}_H.
\end{align*}
\end{small}}
Next, note that the policy of the higher-tier agent is updated toward the future policy estimated by the lower-tier agent, i.e.,  $\bar{\pi}_H \approx f(\bar{\pi}_L)$. We have
{\setlength\abovedisplayskip{3pt}
\setlength\belowdisplayskip{3pt}
\begin{small}
\begin{align*}\nonumber
&\quad V^{\bar{\pi}_L}({\bm{s}}_L;\bar{\pi}_H)\\
&  = \mathbb{E}[R_L({\bm{s}}_L, \bar{\pi}_L({\bm{s}}_L); \tau_H(\bar{\pi}_H)) + \gamma_LV^{\pi_L}(g({\bm{s}}_L, \bar{\pi}_L, \bar{\pi}_H)); \pi_L)\\
& \quad - V^{\pi_L}(s_L;\pi_H) +V^{\pi_L}(s_L;\pi_H)] \\
                     & \approx \mathbb{E}[R_L({\bm{s}}_L, \bar{\pi}_L({\bm{s}}_L); \tau_H(f(\bar{\pi}_L))) + \gamma_LV^{\pi_L}(g({\bm{s}}_L, \bar{\pi}_L, f(\bar{\pi}_L));\\
                     &\quad f(\bar{\pi}_L)) - V^{\pi_L}({\bm{s}}_L;\pi_H)+V^{\pi_L}({\bm{s}}_L;\pi_H)] \\
                     & = \max\limits_{\bar{\pi}_L}[\rho_L(\bar{\pi}_L, f(\bar{\pi}_L)) - \rho_L(\pi_H, \pi_L)] + V^{\pi_L}({\bm{s}}_L;\pi_H)\\
                     & \ge V^{\pi_L}({\bm{s}}_L;\pi_H), \forall {\bm{s}}_L \in {\cal{S}}_L,
\end{align*}
\end{small}}
Therefore, monotonic policy improvements of different agents can be achieved after the policy updating stage.

\section{Proof of Proposition 3}

We first proof the Lipschitz characteristic of $G_H(\cdot)$ and $G_L(\cdot)$. Given lower-tier policy $\pi_L^n$,
{\setlength\abovedisplayskip{0pt}
\setlength\belowdisplayskip{0pt}
\begin{small}
\begin{align*}
&\quad \lVert G_H(x(n), y(n)) -  G_H(x^*, y(n))\rVert \\
& = \lVert \max\limits_{\pi_H}(R_H({\bm{s}}_H, \pi_H(s_H); \tau_L(\pi_L^{n+1})) + \gamma_HV^{\pi^n_H}(h({\bm{s}}_H, \pi_H, \pi^{n+1}_L); \\ &\quad \pi^{n}_L))- \max\limits_{\pi_H}(R_H({\bm{s}}_H, \pi_H(s_H); \tau_L(\pi_L^{n+1})) + \gamma_HV^{\pi_H^*}(h({\bm{s}}_H, \pi_H,\\
&\quad \pi^{n+1}_L); \pi^{n}_L))\rVert \\
& \le \lVert\gamma_H\max\limits_{\pi_H}V^{\pi^n_H}(h({\bm{s}}_H, \pi_H, \pi^{n+1}_L), \pi^{n}_L)) \\
& \quad -  \gamma_H\max\limits_{{\bm{s}}_H \in {\cal{S}}_H}V^{\pi_H^*}(h({\bm{s}}_H, \pi_H^*, \pi^{n+1}_L); \pi^{n}_L)\rVert \\
&\le \gamma_H\max_{{\bm{s}}_H \in {\cal{S}}_H}\lVert V^{\pi_H^{n}}({\bm{s}}_H; \pi^{n}_L) - V^{\pi_H^{*}}({\bm{s}}_H; \pi^n_L)\rVert\\
& = \gamma_H\lVert x(n) - x^{*} \rVert_{\infty}.
\end{align*}
\end{small}}
Similarly, with a fixed higher-tier policy $\pi_H^n$ (namely, $\pi^{n+1}_H = \pi^n_H$),
{\setlength\abovedisplayskip{0pt}
\setlength\belowdisplayskip{0pt}
\begin{small}
\begin{align*}
& \quad \lVert G_L(x(n), y(n)) -  G_L(x(n), y^*)\rVert \\
& = \lVert R_L({\bm{s}}_L, \pi_L^{n+1}({\bm{s}}_L); \tau_H(\pi_H^{n+1})) + \gamma_LV^{\pi^{n}_L}(g({\bm{s}}_L, \pi^{n+1}_H, \pi^{n+1}_L);\\
& \quad \pi^{n}_H) - R_L({\bm{s}}_L, \pi_L^{*}({\bm{s}}_L); \tau_H(\pi_H^{n+1})) - \gamma_LV^{\pi^{*}_L}(g({\bm{s}}_L, \pi^{n+1}_H,\\
& \quad  \pi^{*}_L); \pi^{n}_H) \rVert \\
& \le \lVert \gamma_LV^{\pi^{n}_L}(g({\bm{s}}_L, \pi^{n+1}_H, \pi^{n+1}_L); \pi^{n}_H)\\
&\quad - \gamma_L\max_{{\bm{s}}_L \in {\cal{S}}_L}V^{\pi^{*}_L}(g({\bm{s}}_L, \pi^{n+1}_H, \pi^{*}_L); \pi^{n}_H) \rVert \\
& \le \gamma_L \max_{{\bm{s}}_L \in {\cal{S}}_L}\lVert V^{\pi_L^{n}}({\bm{s}}_L; \pi^{n}_H) - V^{\pi_L^{*}}({\bm{s}}_L; \pi^n_H)\rVert\\
& = \gamma_L\lVert y(n) - y^{*} \rVert_{\infty}.
\end{align*}
\end{small}}
Therefore, $G_H(\cdot)$ and $G_L(\cdot)$ are Lipschitz with regard to the $\infty$-norm (namely, $\lVert \cdot \rVert_{\infty}$). Next, since the rewards defined in \eqref{eq:h_reward} and \eqref{eq:l_reward} are bounded, $\beta_H^{n}$ and $\beta_L^{n}$ as estimated errors of accumulated rewards are also bounded, and thus we have
{\setlength\abovedisplayskip{0pt}
\setlength\belowdisplayskip{0pt}
\begin{small}
\begin{align*}
&\quad \mathbb{E}[\lVert\beta_H^{n+1}\rVert_{\infty}^2|{\cal{F}}_n] = \mathbb{E}[\lVert \max\limits_{\pi_H}\sum_{k = 0}^{\bar{n} - 1}\gamma_H^{k}R_H({\bm{s}}_H^{n+k}, \pi_H; \tau_L(\pi^{n+1}_L))\\
&\quad  - \max\limits_{\pi_H}V^{\pi_H}({\bm{s}}_H; \pi^{n+1}_L)\rVert_{\infty}^2] \le K(1 + \lVert x(n) \rVert_{\infty}^2 + \lVert y(n) \rVert_{\infty}^2), \\
\end{align*}
\end{small}}
and
{\begin{small}
\begin{align*}
& \quad \mathbb{E}[\lVert\beta_L^{n+1}\rVert_{\infty}^2|{\cal{F}}_n]\\
& = \mathbb{E}[\lVert V^{\pi^{n+1}_L}({\bm{s}}_L;\hat{\pi}^{n+1}_H) - V^{\pi^{n+1}_L}({\bm{s}}_L;\pi^{n+1}_H)\rVert_{\infty}^2]\\
& = \mathbb{E}[\lVert(R_L({\bm{s}}_L, \pi_L^{n+1}({\bm{s}}_L); \tau_H(\hat{\pi}^{n+1}_H)) - R_L({\bm{s}}_L, \pi_L^{n+1}({\bm{s}}_L); \\
& \quad \tau_H(\pi^{n+1}_H))) +\gamma_L(V^{\pi^{n+1}_L}(g({\bm{s}}_L, \pi_L^{n+1}, \hat{\pi}^{n+1}_H);\hat{\pi}^{n+1}_H) - V^{\pi^{n+1}_L}(\\
& \quad g({\bm{s}}_L, \pi_L^{n+1}, \pi^{n+1}_H); \pi^{n+1}_H))\rVert_{\infty}^2]\\
& \le K(1 + \lVert x(n) \rVert_{\infty}^2 + \lVert y(n) \rVert_{\infty}^2).
\end{align*}
\end{small}}
Since the Lipschitz characteristics of $G_H(\cdot)$ and $G_L(\cdot)$, $x(n)$ and $y(n)$ can be regarded as contractive sequences. Therefore, the value of $K(1 + \lVert x(n) \rVert_{\infty}^2 + \lVert y(n) \rVert_{\infty}^2)$ with proper $K$ decreases close to $0$ with $n \rightarrow \infty$. In this case, $\mathbb{E}[\beta_H^{n+1}|{\cal{F}}_n] \rightarrow 0$  and $\mathbb{E}[\beta_L^{n+1}|{\cal{F}}_n] \rightarrow 0$. Consequently, $\beta_H^{n}$ and $\beta_L^{n}$ are martingale difference sequences with regard to ${\cal{F}}_n$. Therefore, we have
{\setlength\abovedisplayskip{3pt}
\setlength\belowdisplayskip{3pt}
\begin{small}
\begin{align*}
&\quad \sum\limits_{n=0}^{\infty}a(n)\beta_H^n\\
& < \sqrt{\sum\limits_{n = 0}^{\infty}(\beta_H^n)^2\sum\limits_{n = 0}^{\infty}a(n)^2} <\sqrt{(\frac{\gamma_H^{\bar{n}}R_H^{max}}{1 - \gamma_H})^2\sum\limits_{n = 0}^{\infty}a(n)^2} < \infty,\\
& \sum\limits_{n=0}^{\infty}b(n)\beta_L^n < \sqrt{\sum\limits_{n = 0}^{\infty}(\beta_L^n)^2\sum\limits_{n = 0}^{\infty}b(n)^2} <\sqrt{\frac{(R_L^{max})^2}{1-\gamma_L^2}\sum\limits_{n = 0}^{\infty}b(n)^2} < \infty.
\end{align*}
\end{small}}

\section{Proof of Theorem 1}

At first, according to the definitions in \textbf{Remark 2} and \textbf{Proposition 3}, iterates $\{x(n), y(n)\}$ is following the definition of the two time-scale iterates in \cite{Borkar1997}. Therefore,
{\setlength\abovedisplayskip{3pt}
\setlength\belowdisplayskip{3pt}
\begin{small}
\begin{align*}
&\quad\quad \lVert y^* -y(n + 1) \rVert_{\infty}  =  \lVert {\cal{T}}_Ly^* - {\cal{T}}_Ly(n)\rVert_{\infty} \\
& \le  \lVert y^* + b(n)G_L(x(n), y^*) - b(n)y^* - (y(n) \\
& \quad + b(n)G_L(x(n), y(n))- b(n)y(n)) \rVert_{\infty} + \lVert b(n)\beta_L^{n+1} \rVert_{\infty}\\
& \le (1 - b(n))\lVert y^* - y(n) \rVert_{\infty} + b(n)\lVert G_L(x(n), y^*)\\
&\quad - G_L(x(n), y(n)) \rVert_{\infty} + \lVert b(n)\beta_L^{n+1} \rVert_{\infty}\\
& \le (1 - b(n)) \lVert y^* - y(n) \rVert_{\infty} + b(n)(\gamma_L\lVert  y^* - y(n)  \rVert_{\infty} + \lVert\beta_L^{n+1} \rVert_{\infty}).
\end{align*}
\end{small}}
where the error term $b(n)(\gamma_L\lVert  y^* - y(n)  \rVert_{\infty} + \lVert\beta_L^{n+1} \rVert_{\infty})$ can be negligible with $b(n) \to 0$. In this case, there exits a unique global asymptotically stable equilibrium. Additionally, since $G_L(\cdot)$ has been proved to be Lipschitz, $\lambda(\cdot) = \max(G_L(\cdot))$ should also be Lipschitz with regard to the $\infty$-norm. Thus, for each $x(n) \varpropto \pi_H^n$, $\dot{y}(n) = G_L(x(n), y(n))$ has a unique global asymptotically stable equilibrium $y(n) = \lambda(x(n))$ such that $\lambda(\cdot)$ is Lipschitz, and this assumption is satisfied.

Next, with updated $y(n+1) = \lambda(x(n))$, $G_H(x(n), \lambda(x(n)))$ has a deterministic improvement incurred by the $\lambda(x(n))$, and we have
{\setlength\abovedisplayskip{3pt}
\setlength\belowdisplayskip{3pt}
\begin{small}
\begin{align*}
&\quad\quad   \lVert x^* - x(n + 1) \rVert_{\infty} =  \lVert {{\cal{T}}_Hx^* - \cal{T}}_Hx(n)\rVert_{\infty} \\
& \le  \lVert x^* + a(n)G_H(x^*, y(n)) - a(n)x^* - (x(n)\\
& \quad + a(n)G_H(x(n), y(n)) - a(n)x(n))\rVert_{\infty} + \lVert a(n)\beta_H^{n+1} \rVert_{\infty} \\
&\le (1 - a(n))\lVert x^* - x(n) \rVert_{\infty} + a(n)\lVert G_H(x(n), y(n)) \\
&\quad - G_H(x^*, y(n))\rVert_{\infty} + \lVert a(n)\beta_H^{n+1} \rVert_{\infty} \\
&\le (1 - a(n)) \lVert x^* - x(n) \rVert_{\infty} +  a(n)(\gamma_H\lVert x^* - x(n) \rVert_{\infty} + \lVert \beta_H^{n+1} \rVert_{\infty}).
\end{align*}
\end{small}}
where the error term $a(n)(\gamma_H\lVert x^* - x(n) \rVert_{\infty} + \lVert \beta_H^{n+1} \rVert_{\infty})$ can be negligible with $a(n) \to 0$. In this case, $x(n)$ can be updated through iterations towards the unique global asymptotically stable equilibrium. Therefore, $\dot{x}(n) = G_H(x(n), \lambda(x(n)))$ has a unique global asymptotically stable equilibrium $x^*$, and \textbf{Assumption 2} is satisfied. Finally, since the rewards defined in \eqref{eq:h_reward} and \eqref{eq:l_reward} are bounded, $x(n)$ or $y(n)$ equivalent to corresponding discounted accumulated reward is also bounded. Therefore, \textbf{Assumption 3} is satisfied. Then, according to \textbf{Lemma 1}, iterates $\{x(n), y(n)\}$ converge to the optimum. Since the consistency between policy and value function, policies converges to the optimum, and the proof is completed.

\section{Proof of Theorem 2}

In this proof, we first define the optimal Bellman operators as
{\setlength\abovedisplayskip{3pt}
\setlength\belowdisplayskip{3pt}
\begin{small}
\begin{align*}
 {\cal{T}}_HV^{\pi_H}({\bm{s}}_H; \pi_L) & = \max\limits_{\pi_H}(R_H({\bm{s}}_H, \pi_H({\bm{s}}_H); \tau_L(\pi_L^{*}))\\
  &\quad + \gamma_HV^{\pi_H}(h({\bm{s}}_H, \pi_H, \pi_L^{*}); \pi_L^{*})),\\
{\cal{T}}_LV^{\pi_L}({\bm{s}}_L; \pi_H) & = \max\limits_{\pi_L}(R_L({\bm{s}}_L, \pi_L({\bm{s}}_L); \tau_H(\pi_H^{*})) \\
  &\quad + \gamma_LV^{\pi_L}(h({\bm{s}}_L, \pi_L, \pi_H^{*}); \pi_H^{*})).
\end{align*}
\end{small}}
Then, we define empirical Bellman operators $\hat{\cal{T}}_H$ and $\hat{\cal{T}}_L$, i.e.,
{\setlength\abovedisplayskip{3pt}
\setlength\belowdisplayskip{3pt}
\begin{small}
\begin{align*}
\hat{\cal{T}}^n_HV^{\pi_H}({\bm{s}}_H; \pi_L) & = \max\limits_{\pi_H}(R_H({\bm{s}}_H, \pi_H({\bm{s}}_H); \tau_L(\pi^n_L)) \\
  &\quad + \gamma_HV^{\pi^n_H}(h({\bm{s}}_H, \pi_H, \pi^n_L); \pi^n_L))\\
\hat{\cal{T}}^n_LV^{\pi_L}({\bm{s}}_L; \pi_H) &= \max\limits_{\pi_L}(R_L({\bm{s}}_L, \pi_L({\bm{s}}_L); \tau_H(\pi^n_H)) \\
  &\quad + \gamma_LV^{\pi^n_L}(h({\bm{s}}_L, \pi_L, \pi^n_H); \pi^n_H)).
\end{align*}
\end{small}}
Therefore, sequences in \eqref{eq:h_tier} and \eqref{eq:l_tier} can be rewritten as
{\setlength\abovedisplayskip{3pt}
\setlength\belowdisplayskip{3pt}
\begin{small}
\begin{align*}
x(n+1) &= (1 - a(n))x(n) + a(n)\hat{\cal{T}}^n_Hx(n) + a(n)\beta_H^{n+1}\\
y(n+1) &= (1 - b(n))y(n) + b(n)\hat{\cal{T}}^n_Ly(n) + b(n)\beta_L^{n+1}.
\end{align*}
\end{small}}
Since the rewards defined in \eqref{eq:h_reward} and \eqref{eq:l_reward} are bounded,  $V^{\pi_H}(\cdot)$ and $V^{\pi_L}(\cdot)$ are bounded by $\frac{R_H^{max}}{1- \gamma_H}$ and $\frac{R_L^{max}}{1 - \gamma_L}$, respectively. If we denote $e_L^{n} = \hat{\cal{T}}^n_L\hat{y}^{n} - {\cal{T}}_L^{n}y^*$ with $\hat{y}^{0} = y^{*}$,  we have $\hat{y} = \frac{1}{n}\sum_{i = 0}^{n-1}({\hat{\cal{T}}^i_L\hat{y}^{i}} + \beta_L^{i+1}) = \frac{1}{n}(\sum_{i = 0}^{n-1}{\cal{T}}^i_Ly^{*} + \sum_{i = 0}^{n-1}\beta_L^{i+1} + \sum_{i = 0}^{n-1}e_L^i)$, and
{\setlength\abovedisplayskip{3pt}
 \setlength\belowdisplayskip{3pt}
\begin{small}
\begin{align*}
&\quad \lVert y^{*} - \hat{y} \rVert_{\infty}\\
& = \lVert \frac{1}{n}\sum\limits_{i = 0}^{n}{y^{*}} - \frac{1}{n}(\sum\limits_{i = 0}^{n-1}{\cal{T}}_L^{i}y^{*} + \sum\limits_{i = 0}^{n-1}\beta_L^{i+1} + \sum\limits_{i = 0}^{n-1}e_L^i) \rVert_{\infty} \\
& \le \frac{1}{n}\sum\limits_{i = 0}^{n-1}\lVert {\cal{T}}^{i+1}_L y^{*} - {\cal{T}}_L^{i}y^{*} \rVert_{\infty} + \frac{1}{n}\lVert \sum\limits_{i = 0}^{n-1}\beta_L^{i+1} + \sum\limits_{i = 0}^{n-1}e_L^i \rVert_{\infty} \\
& \le \frac{1}{n}\sum\limits_{i = 0}^{n-1}{\gamma_L}^{i} \lVert {\cal{T}}_Ly^{*} -y^{*} \rVert_{\infty} + \frac{1}{n}\lVert \sum\limits_{i = 0}^{n-1}\beta_L^{i+1} + \sum\limits_{i = 0}^{n-1}e_L^i \rVert_{\infty} \\
& \le \frac{1}{n}\sum\limits_{i = 0}^{n-1}{\gamma_L}^{i}2V_L^{max} + \frac{1}{n}\lVert \sum\limits_{i = 0}^{n-1}\beta_L^{i+1} + \sum\limits_{i = 0}^{n-1}e_L^i \rVert_{\infty}\\
&  = \frac{2R_L^{max}(1 - {\gamma_L}^n)}{n(1 - \gamma_L)^2} + \frac{1}{n}\lVert \sum\limits_{i = 0}^{n-1}\beta_L^{i+1} + \sum\limits_{i = 0}^{n-1}e_L^i \rVert_{\infty}\\
& \le \frac{2R_L^{max}}{n(1 - \gamma_L)^2} + \frac{1}{n}\lVert \sum\limits_{i = 0}^{n-1}\beta_L^{i+1} + \sum\limits_{i = 0}^{n-1}e_L^i \rVert_{\infty}.
\end{align*}
\end{small}}

Therefore, $\lVert y^* - y(n) \rVert_{\infty} =  \lVert y^* - {\hat{\cal{T}}_L}^{n}y \rVert_{\infty}  \le  \lVert y^* - \hat{y} \rVert_{\infty} + \lVert \hat{y} - {\hat{\cal{T}}_L}^{n}y \rVert_{\infty}  = 2(\frac{2R_L^{max}}{n(1 - \gamma_L)^2} + \frac{1}{n}\lVert \sum_{i = 0}^{n-1}\beta_L^{i+1} + \sum_{i = 0}^{n-1}e_L^i \rVert_{\infty})  = \frac{4R_L^{max}}{n(1 - \gamma_L)^2} + \frac{2}{n}\lVert \sum_{i = 0}^{n-1}\beta_L^{i+1} + \sum_{i = 0}^{n-1}e_L^i \rVert_{\infty}$.
According to Hoeffding-Azuma's inequality \cite{Amiri2019} and a union bound of the decision space, the probability of failure in any condition is smaller than $\mathbb{P}(\lVert \sum_{i = 0}^{n-1}\beta_L^{i+1} + \sum_{i = 0}^{n-1}e_L^i \rVert_{\infty} \ge \epsilon) \le 2\lvert{\cal{S}}_L\rvert\lvert{\cal{A}}_L\rvert\cdot\exp(- \frac{\epsilon^2(1 - \gamma_L)^{2}}{32n{R_L^{max}}^2})$. Similarly, $\hat{x}$ can also be approximated as a sum of past estimations, i.e., $ \hat{x} = \frac{1}{n}(\sum_{i = 0}^{n-1}{\cal{T}}^i_Hx^{*} + \sum_{i = 0}^{n-1}\beta_H^{i+1} + \sum_{i = 0}^{n-1}e_H^i)$, and the probability of failure can be given by $\mathbb{P}(\lVert \sum_{i = 0}^{n-1}\beta_H^{i+1} + \sum_{i = 0}^{n-1}e_H^i \rVert_{\infty} \ge \epsilon) \le 2\lvert{\cal{S}}_H\rvert\lvert{\cal{A}}_H\rvert\cdot\exp(- \frac{\epsilon^2(1 - \gamma_H)^{2}}{32n{R_H^{max}}^2})$.
Set the error rate less than $\delta$. Hence, $\mathbb{P}(\frac{2}{n}\lVert \sum_{i = 0}^{n-1}\beta_L^{i+1} + \sum_{i = 0}^{n-1}e_L^i \rVert_{\infty} \le 8\sqrt{\frac{2{R_L^{max}}^2}{n(1 - \gamma_L)^2}\ln\frac{2\lvert{\cal{S}}_L\rvert\lvert{\cal{A}}_L\rvert}{\delta}} ) \ge 1 - \delta$ and $\mathbb{P}(\frac{2}{n}\lVert \sum_{i = 0}^{n-1}\beta_H^{i+1} + \sum_{i = 0}^{n-1}e_H^i \rVert_{\infty} \le 8\sqrt{\frac{2{R_H^{max}}^2}{n(1 - \gamma_H)^2}\ln\frac{2\lvert{\cal{S}}_H\rvert\lvert{\cal{A}}_H\rvert}{\delta}} ) \ge 1 - \delta$. Therefore, we can derive \eqref{eq:x_error} and \eqref{eq:y_error}.
\end{appendices}

% Generated by IEEEtran.bst, version: 1.14 (2015/08/26)

\end{document}